\documentclass[11pt,lettersize]{article}
\usepackage{graphicx,color}
\usepackage{amsmath,amssymb,dsfont,enumerate,fullpage,epsfig,multirow,longtable}
\usepackage{epsfig,epstopdf,hhline}
\usepackage{cases}
\usepackage{algorithm}
\usepackage[noend]{algorithmic}

\usepackage[numbers]{natbib}
\usepackage[top=1 in, bottom=1 in, left=1 in, right=1 in, letterpaper]{geometry}
\usepackage{tikz}
\usetikzlibrary{arrows,positioning,decorations.pathreplacing,shapes}

\newenvironment{proof}{\noindent\emph{Proof\ }}{\hspace*{\fill}$\Box$\medskip}

\newtheorem{theorem}{Theorem}
\newtheorem{definition}{Definition}
\newtheorem{lemma}{Lemma}

\newtheorem{proposition}{Proposition}

\usepackage{amsmath}
\usepackage{paralist}
\usepackage{framed}

\newcommand\restr[2]{{
  \left.\kern-\nulldelimiterspace 
  #1 
  \vphantom{\big|} 
  \right|_{#2} 
  }}

\newcommand{\vect}[1]{\ensuremath{\mathbf{#1}}}

\title{Online Primal-Dual Algorithms with Configuration Linear Programs}

\author{
Nguyen Kim Thang\thanks{Research supported by the ANR project OATA n\textsuperscript{o} ANR-15-CE40-0015-01}\\
IBISC, University Paris-Saclay, France}

\date{}

\begin{document}

\maketitle

\begin{abstract}
In this paper, we present primal-dual approaches based on configuration linear programs to 
design competitive online algorithms for problems with arbitrarily-grown objective. 
Non-linear, especially convex, objective functions have been extensively studied in recent years 
in which approaches relies crucially on the convexity property of cost functions. 
Besides, configuration linear programs have been considered typically in offline setting 
and the main approaches are rounding schemes. 

In our framework, we consider configuration linear programs coupled with a primal-dual approach.
This approach is particularly appropriate for non-linear (non-convex)
objectives in online setting. By the approach, we first present a simple greedy algorithm
for a general cost-minimization problem. 
The competitive ratio of the algorithm is characterized by the mean of a notion, called \emph{smoothness}, 
which is inspired by a similar concept in the context of algorithmic game theory.
The algorithm gives \emph{optimal} (up to a constant factor) competitive ratios 
while applying to different contexts such as network routing, vector scheduling, energy-efficient scheduling
and non-convex facility location.


Next, we consider the online $0-1$ covering problems with non-convex objective.   
Building upon the resilient ideas from the primal-dual framework with configuration LPs,
we derive a competitive algorithm for these problems. Our result generalizes the 
online primal-dual algorithm developed recently by \citet{AzarBuchbinder16:Online-Algorithms} 
for convex objectives with monotone gradients to non-convex objectives. 
The competitive ratio is now characterized by a new concept, 
called \emph{local smoothness} --- a notion inspired by the smoothness. 
Our algorithm yields \emph{tight} competitive ratio for the objectives such as the sum of $\ell_{k}$-norms and 
gives competitive solutions for online problems of submodular minimization and some natural non-convex minimization 
under covering constraints.  
\end{abstract}

\thispagestyle{empty}

\newpage

\setcounter{page}{1}

\section{Introduction}

In the paper, we consider problems of minimizing the total cost of resources used to satisfy online requests.
One phenomenon, known as \emph{economy of scale}, is that the cost grows sub-linearly with the amount of resources
used. That happens in many applications in which one gets a discount when buying resources in bulk. A representative 
setting is the extensively-studied domain of sub-modular optimization. Another phenomenon, known as \emph{diseconomy of scale},
is that the cost grows super-linearly on the quantity of used resources. An illustrative example for this phenomenon is 
the energy cost in computation where the cost grows super-linearly, typically as a convex function. 
The phenomenon of diseconomy of scale has been 
widely studied in the domain of convex optimization \cite{BoydVandenberghe04:Convex-Optimization}.  
Non-convex objective functions appears in various problems, ranging from scheduling, sensor energy
management, to influence and revenue maximization, and facility location.
For example, in scheduling of malleable jobs on parallel machines, 
the cost grows as a non-convex function \cite{Hunold15:One-step-toward} which is due to the 
parallelization and the synchronization. Besides, in practical aspect of facility location, 
the facility costs to serve clients are rarely constant or simply a convex function of the number of clients. 
Apart of some fixed opening amount, the cost would initially increase fast until some threshold on the number of clients, 
then becomes more stable before quickly increases again as the number of clients augments. 
This behaviour of cost functions widely happens in economic contexts. Such situations raises the demand 
of designing algorithms with performance guarantee for non-convex objective functions. 
In this paper, we consider problems in which the cost grows \emph{arbitrarily} with the amount of used resources.  
 
\subsection{A General Problem and Primal-Dual Approach}  	\label{sec:intro-gen}
 
\paragraph{General Problem.}
In the problem, there is a set of resources $\mathcal{E}$ and requests arrive online. 
At the arrival of request $i$, a set of feasible strategies (actions) $\mathcal{S}_{i}$ to satisfy request $i$ is revealed. 
Each strategy $s_{ij} \in \mathcal{S}_{i}$ consists of a subset of resources in $\mathcal{E}$.
Each resource $e$ is associated to a 
non-negative non-decreasing \emph{arbitrary} cost function $f_{e}$ and the cost induced by resource $e$ 
depending on the set of requests using $e$. The cost of a solution is the total cost of resources, i.e., 
$\sum_{e} f_{e}(A_{e})$ where $A_{e}$ is the set of requests using resource $e$.  
The goal is design an algorithm that upon the arrival of each request, selects a feasible strategy for the request
while maintaining the cost of the overall solution as small as possible. We consider the standard
\emph{competitive ratio} as the performance measure of an algorithm. Specifically, 
an algorithm is $r$-competitive if for any instance, the ratio between the cost of the algorithm and that of 
an optimal solution is at most $r$.

\paragraph{Primal-Dual Approach.} 
We consider an approach based on linear programming for the problem. 
The first crucial step for any LP-based approach is to derive a LP formulation with reasonable
\emph{integrality gap}, which is defined as the ratio between the 
optimal integer solution of the formulation and the optimal solution without the integer condition. 
As the cost functions are non-linear, it is not surprising that the natural relaxation 
suffers from large integrality gap. This issue has been observed and resolved by  
\citet{MakarychevSviridenko14:Solving-optimization}.
\citet{MakarychevSviridenko14:Solving-optimization} considered an offline variant of the problem in which
the resource cost functions are convex. 
They systematically strengthen the natural formulations by 
introducing an exponential number of new variables and new constraints connecting 
new variables to original ones. Consequently, the new formulation, in form of a configuration LP, 
significantly reduces the integrality gap. Although there are exponentially number of variables, 
\citeauthor{MakarychevSviridenko14:Solving-optimization} showed that a fractional $(1+\epsilon)$-approximatly 
optimal solution of the configuration LP can be computed in polynomial time. 
Then, by rounding the fractional solution, the authors derived an $B_{\alpha}$-approximation algorithm 
for the resource cost minimization problem in which all cost functions are polynomial of degree at most $\alpha$. 
Here $B_{\alpha}$ denotes the Bell number and asymptotically $B_{\alpha} = \Theta\bigl((\alpha/\log \alpha)^{\alpha}\bigr)$.


The rounding scheme in \cite{MakarychevSviridenko14:Solving-optimization} is 
\emph{intrinsically offline} and it is not suitable in online setting. Moreover, another issue in the problem is that cost functions are not 
necessarily convex. That represents a substantial obstacle since all currently known techniques for non-linear 
objectives relies crucially on the convexity of cost functions and Fenchel duality
\cite{DevanurHuang14:Primal-Dual,AzarDevanur15:Speed-Scaling,MenacheSingh15:Online-Caching,GuptaKrishnaswamy12:Online-Primal-Dual,HuangKim15:Welfare-maximization,DevanurJain12:Online-matching,AzarBuchbinder16:Online-Algorithms}.  

To overcome these difficulties, we consider a primal-dual approach with configuration LPs. 
First, primal-dual is particularly appropriate since one does not 
have to compute an optimal fractional solution that needs the full information on the instance. 
Second, in our approach, the dual variables of the configuration LP 
have intuitive meanings and the dual constraints indeed guide the decisions of the algorithm. 
The key step in the approach is to show that the constructed dual variables constitute a dual feasible 
solution. In order to prove the dual feasibility, we define a notion of \emph{smoothness} of 
functions. This definition is inspired by the smoothness framework introduced by \citet{Roughgarden15:Intrinsic-Robustness} 
in the context of algorithmic game theory to characterize the price of anarchy for large classes of games. 
The smoothness notion allows us not only to prove the dual feasibility but also to establish the competitiveness of 
algorithms in our approach.
We characterize the performance of algorithms
using the notion of smoothness in a similar way as the price of anarchy characterized by the smoothness 
argument \cite{Roughgarden15:Intrinsic-Robustness}. Through this notion, we show an interesting connection between 
online algorithms and algorithmic game theory. 

\begin{definition}
Let $\mathcal{N}$ be a set of requests.
A set function $f: 2^{\mathcal{N}} \rightarrow \mathbb{R}^{+}$ is $(\lambda,\mu)$-\emph{smooth}
if for any set $A = \{a_{1}, \ldots, a_{n}\} \subseteq \mathcal{N}$ and any collection 
$B_{1} \subseteq B_{2} \subseteq \ldots \subseteq B_{n} \subseteq B \subseteq \mathcal{N}$, 
the following inequality holds.
$$
\sum_{i = 1}^{n} \left[ f\bigl( B_{i} \cup a_{i} \bigr) - f\bigl( B_{i} \bigr)\right]
\leq \lambda f\bigl( A \bigr) + \mu f\bigl( B \bigr)
$$
A set of cost functions $\{f_{e}: e \in \mathcal{E}\}$ is $(\lambda,\mu)$-\emph{smooth} if every function 
$f_{e}$ is $(\lambda,\mu)$-\emph{smooth}.
\end{definition}
 
Intuitively, given a $(\lambda,\mu)$-smooth function, the quantity $\frac{\lambda}{1-\mu}$ measures how far 
the function is from being linear. If a function is linear then it is $(1,0)$-smooth.
 
\begin{theorem}
Assume that all resource cost functions are $(\lambda,\mu)$-smooth for some parameters $\lambda > 0$,
$\mu < 1$. Then there exists a greedy $\frac{\lambda}{1-\mu}$-competitive algorithm for the 
general problem. 
\end{theorem}
 
Note that, restricted to the class of polynomials with non-negative coefficients, our algorithm yields 
the competitive ratio of $O(\alpha^{\alpha})$ (consequence of Lemma \ref{lem:smooth-convex}) while the best-known 
approximation ratio is $B_{\alpha} \approx \bigl(\frac{\alpha}{\log \alpha}\bigr)^{\alpha}$ \cite{MakarychevSviridenko14:Solving-optimization}.
However, our greedy algorithm is light-weight and much simpler than that in \cite{MakarychevSviridenko14:Solving-optimization} 
which involves in solving an LP of exponential size and rounding fractional solutions. 
Hence, our algorithm can also be used to design approximation algorithms if one looks for the tradeoff between the simplicity and the performance guarantee.  
 
\paragraph{Applications.} 
We show the applicability of the theorem by deriving competitive algorithms for several problems in online setting, 
such as  \textsc{Minimum Power Survival Network Routing}, \textsc{Vector Scheduling}, 
\textsc{Energy-Efficient Scheduling}, \textsc{Prize Collecting Energy-Efficient Scheduling}, 
\textsc{Non-Convex Facility Location}. Among such applications, 
the most representative ones are the \textsc{Energy-Efficient Scheduling} problem and 
the \textsc{Non-Convex Facility Location} problem.

In \textsc{Online Energy-Efficient Scheduling}, one has to process jobs on unrelated machines
without migration with the objective of minimizing the total energy. No result has been known for this problem in multiple machine
environments. Among others, a difficulty is the construction of formulation with bounded integrality gap. 
We notice that for this problem, \citet{GuptaKrishnaswamy12:Online-Primal-Dual} gave a primal-dual 
competitive algorithm for a single machine. 
However, their approach cannot be used for unrelated machines due to the large integrality gap 
of their formulation. For these problems, we present competitive algorithms with \emph{arbitrary} cost functions beyond the convexity property. 
Note that the convexity of cost functions is a crucial property employed in the analyses of previous work. 
If the cost functions have typical form $f(x) = x^{\alpha}$ then the competitive ratio 
is $O(\alpha^{\alpha})$ and this is optimal up to a constant factor for all the problems above. 
Besides, in offline setting, this ratio is close to the currently best-known 
approximation ratio $B_{\alpha} \approx \bigl(\frac{\alpha}{\log \alpha}\bigr)^{\alpha}$ \cite{MakarychevSviridenko14:Solving-optimization}.

In \textsc{Online Non-Convex Facility Location}, clients arrive online and have to be assigned to facilities.
The cost of a facility consists of a fixed opening cost and and a serving cost, which is an arbitrary monotone function 
depending on the number of clients assigned to the facility. The objective is to minimize the total client-facility 
connection cost and the facility cost. This problem is related to the capacitated network design and energy-efficient routing problems \cite{AntoniadisIm14:Hallucination-Helps:,KrishnaswamyNagarajan14:Cluster-before}.
In the latter, given a graphs and a set of connectivity demands, the cost of each edge (node) is \emph{uniform} and given by 
$c + f^{\alpha}$ where $c$ is a fixed cost for every edge and $f$ is the total of flow passed through the edge (node). 
(Here uniformity means the cost functions are the same for every edge.) The objective is to minimize the total cost while 
satisfying all connectivity demands. \citet{AntoniadisIm14:Hallucination-Helps:,KrishnaswamyNagarajan14:Cluster-before} have provided online/offline algorithms with poly-logarithmic guarantees. 
It is an intriguing open questions (originally raised in \cite{AndrewsAntonakopoulos10:Minimum-Cost-Network}) to design a poly-logarithmic competitive algorithm for non-uniform cost functions.  
The \textsc{Online Non-Convex Facility Location} can be seen as a step towards this goal. 
In fact, the former can be considered as the connectivity problem on a simple depth-2-graph but the cost functions 
are now non-uniform. 
 
This problem is beyond the scope of general problem but we show that the resilient ideas from
the primal-dual framework can be used to derive competitive algorithm. Specifically,  
we present a $O(\log n + \frac{\lambda}{1-\mu})$-competitive algorithm if the cost function is $(\lambda,\mu)$-smooth. 
The algorithm is inspired by the Fortakis primal-dual algorithm in classic setting \cite{Fotakis07:A-primal-dual-algorithm}
and our primal-dual approach based on configuration LPs. In particular, for the problem 
with non-uniform cost functions such as $c_{i} + w_{i} f_{i}^{\alpha}$ 
where $c_{i},w_{i}$ are parameters depending on facility $i$ and $f_{i}$ is the number of clients assigned to facility $i$,
the algorithm yields a competitive ratio of $O(\log n + \alpha^{\alpha})$.
 
\subsection{Primal-Dual Approach for $0-1$ Covering Problems}  
 
\paragraph{$0-1$ Covering Problems.}
We consider an extension of the general problem described in the previous section
in which the resources are subject to covering constraints. 
Formally, let $\mathcal{E}$ be a set of $n$ resources 
and let $f: \{0,1\}^{n} \rightarrow \mathbb{R}^{+}$ be an  \emph{abitrary} monotone cost function.
Let $x_{e} \in \{0,1\}$ be a variable indicating whether resource $e$ is selected. 
The covering constraints $\sum_{e} b_{i,e} x_{e} \geq 1$ for every $i$ are revealed one-by-one and at any 
step, one needs to maintain a feasible integer solution $\vect{x}$.
The goal is to design an algorithm that minimizes $f(\vect{x})$ subject to the online covering constraints 
and $x_{e} \in \{0,1\}$ for every $e$. 

Very recently, \citet{AzarBuchbinder16:Online-Algorithms} have presented a general primal-dual framework 
when function $f$ is convex with monotone gradient. 
The framework, inspired by the Buchbinder-Naor framework \cite{BuchbinderNaor09:The-Design-of-Competitive}
for linear objectives,  
crucially relies on Fenchel duality and the convexity of the objective functions. 
We overcome this obstacle for non-convex functions and also for convex functions with non-monotone gradients 
by considering configuration LP
corresponding to the problem and multilinear extension of function $f$. 
Given $f: \{0,1\}^{n} \rightarrow \mathbb{R}^{+}$, its multilinear extension $F: [0,1]^{n} \rightarrow \mathbb{R}^{+}$ 
is defined as $F(\vect{x}) := \sum_{S} f(\vect{1}_{S}) \prod_{e \in S} x_{e} \prod_{e \notin S} (1 - x_{e})$
where $\vect{1}_{S}$ is the characteristic vector of $S$ (i.e., the $e^{\textnormal{th}}$-component of 
$\vect{1}_{S}$ equals $1$ if $e \in S$ and equals 0 otherwise). 
Building upon the primal-dual framework in \cite{AzarBuchbinder16:Online-Algorithms,BuchbinderNaor09:The-Design-of-Competitive} 
and the resilient ideas due to 
the primal-dual approach for the general problem described earlier, we present a competitive algorithm,
which follows closely to the one in \cite{AzarBuchbinder16:Online-Algorithms}, 
for the fractional $0-1$ covering problem. 
Specifically, we introduce the notion of locally-smooth and 
characterize the competitive ratio using these local smoothness' parameters.

\begin{definition}	\label{def:local-smooth}
Let $\mathcal{E}$ be a set of $n$ resources.
A differentiable function $F: [0,1]^{n} \rightarrow \mathbb{R}^{+}$ is $(\lambda,\mu)$-\emph{locally-smooth}
if for any set $S \subset \mathcal{E}$, the following inequality holds.
\begin{equation}	\label{eq:local-smooth}
\sum_{e \in S} \nabla_{e} F(\vect{x}) \leq \lambda F\bigl( \vect{1}_{S} \bigr) + \mu F\bigl( \vect{x} \bigr)
\end{equation}
\end{definition}

\begin{theorem}
Let $F$ be the multilinear extension of the objective cost $f$ and 
$d$ be the maximal row sparsity of the constraint matrix, i.e., $d = \max_{i} |\{b_{ie}: b_{ie} > 0\}|$.  
Assume that $F$ is $\bigl(\lambda, \frac{\mu}{\ln(1+2d^{2})}\bigr)$-locally-smooth 
for some parameters $\lambda > 0$ and $\mu < 1$. 
Then there exists a $O\bigl( \frac{\lambda}{1 - \mu}  \cdot \ln d \bigr)$-competitive 
algorithm for the fractional covering problem. 
\end{theorem}

Our algorithm, as well as the algorithm in \cite{AzarBuchbinder16:Online-Algorithms} for convex with monotone gradients 
and the recent algorithm for $\ell_{k}$-norms \cite{NagarajanShen17:Online-Covering}, are natural extensions of the 
Buchbinder-Naor primal-dual framework \cite{BuchbinderNaor09:The-Design-of-Competitive}. A distinguishing point of 
our algorithm compared to the ones in \cite{AzarBuchbinder16:Online-Algorithms,NagarajanShen17:Online-Covering} is 
that in the latter, the gradient $\nabla f(\vect{x})$ at the current primal solution $\vect{x}$ is used to define a multiplicative 
update for the primal whereas we use the gradient of the multilinear extension $\nabla F(\vect{x})$ to define such update. 
This (rather small) modification, coupling with configuration LPs, enable us to derive a competitive algorithm for
convex objective functions whose gradients are not necessarily monotone and more generally, for non-convex objectives.  
Moreover, the use of configuration LPs and the notion of local smoothness is twofold: (i) it avoids the cumbersome 
technical details in the analysis as well as in the assumptions of objective functions; (ii) it reduces the proof of bounding 
the competitive ratios for classes of objective functions to determining the local-smoothness parameters.

Specifically, we apply our algorithm to several widely-studied classes of functions in optimization. 
First, for the class of non-negative polynomials of degree $k$, the algorithm yields a $O\bigl((k\log d)^{k} \bigr)$-competitive 
fractional solution that matches to a result in \cite{AzarBuchbinder16:Online-Algorithms}. 
Second, for the class of sum of $\ell_{k}$-norms, recently \citet{NagarajanShen17:Online-Covering}, 
based on the algorithm in \cite{AzarBuchbinder16:Online-Algorithms}, have 
presented a nearly tight $O(\log d + \log \frac{\max a_{ij}}{\min a_{ij}})$-competitive algorithm where
$a_{ij}$'s are entries in the covering matrix. We show that our algorithm yields a \emph{tight} 
$O(\log d)$-competitive ratio for this class of functions. (The matching lower bound is given in \cite{BuchbinderNaor09:Online-primal-dual}.)  
Third, beyond convexity, we consider a natural class of non-convex cost functions which represent a typical behaviour of resources in 
serving demand requests.  
Non-convexity represents a strong barrier in optimization in general and in the design of algorithms in particular. 
We show that our algorithm is competitive for this class of functions. 
Finally, we illustrate the applicability of our algorithm to the class of submodular functions. 
We make a connection between the local-smooth parameters to the concept of \emph{total curvature} $\kappa$ of submodular functions.
The total curvature has been widely used to determines both upper and lower bounds on the approximation ratios
for many submodular and learning problems \cite{ConfortiCornuejols84:Submodular-set-functions,GoemansHarvey09:Approximating-submodular,BalcanHarvey12:Learning-Submodular,Vondrak10:Submodularity-and-Curvature:,IyerJegelka13:Curvature-and-optimal,SviridenkoVondrak17:Optimal-approximation}. We show that our algorithm yields a 
$O\bigl(\frac{\log d}{1 - \kappa}\bigr)$-competitive fractional solution for the problem of minimizing a submodular function 
under covering constraints. To the best of our knowledge, the submodular minimization under general convering constraints has not been 
studied in the online computation setting. 

\subsection{Related work}
In this section we summarize related work to our approach. Each problem, together with its related work, in the applications of 
the main theorems is formally given in the corresponding section.  
%

In this paper, we systematically strengthen natural LPs by the construction of the configuration LPs 
presented in \cite{MakarychevSviridenko14:Solving-optimization}.
\citet{MakarychevSviridenko14:Solving-optimization} propose a scheme that consists of solving 
the new LPs (with exponential number of variables) 
and rounding the fractional solutions to integer ones using decoupling inequalities. By this method, 
they derive approximation algorithms for several (offline) optimization problems which can formulated by 
linear constraints and objective function as a power of some constant $\alpha$. 
Specifically, the approximation ratio is proved to be the Bell number $B_{\alpha}$ for several problems.
   
In our approach, a crucial element to characterize the performance of an algorithm is the smoothness
property of functions. The smooth argument is introduced by \citet{Roughgarden15:Intrinsic-Robustness} in the context of algorithmic game theory
and it has successfully characterized the performance of equilibria (price of anarchy) in many classes of 
games such as congestion games, etc \cite{Roughgarden15:Intrinsic-Robustness}. 
This notion inspires the definition of smoothness in our paper.

Primal-dual methods have been shown to be powerful tools in online computation.
\citet{BuchbinderNaor09:The-Design-of-Competitive}
presented a primal-dual method for linear programs with packing/covering constraints.
Their method unifies several previous potential-function-based analyses and give a 
principled approach to design and analyze algorithms for problems with linear relaxations.
Convex objective functions have been extensively studied in online settings in recent years, 
in areas such as energy-efficient scheduling \cite{AnandGarg12:Resource-augmentation,Thang13:Lagrangian-Duality,DevanurHuang14:Primal-Dual,ImKulkarni14:SELFISHMIGRATE:-A-Scalable,AzarDevanur15:Speed-Scaling}, 
paging \cite{MenacheSingh15:Online-Caching}, network routing \cite{GuptaKrishnaswamy12:Online-Primal-Dual},
combinatorial auctions \cite{BlumGupta11:Welfare-and-profit,HuangKim15:Welfare-maximization}, matching \cite{DevanurJain12:Online-matching}.  
Recently, \citet{AzarBuchbinder16:Online-Algorithms} gave an unified framework for covering/packing problems with 
convex objectives whose gradients are monotone. Consequently, improved algorithms have been derived for several problems.
The crucial point in the design and analysis in the above approaches relies on the convexity of 
cost functions. Specifically, the construction of dual programs is based on convex conjugates 
and Fenchel duality for primal convex programs. 
Very recently, \citet{NagarajanShen17:Online-Covering}
have considered objective functions as the of sum of $\ell_{k}$-norms. 
This class of functions do not fall into the framework 
developped in \cite{AzarBuchbinder16:Online-Algorithms} since the gradients are not necessarily monotone. 
By a different analysis, \citet{NagarajanShen17:Online-Covering} proved that the algorithm presented in
\cite{AzarBuchbinder16:Online-Algorithms} yields a nearly tight $O\bigl(\log d + \log \frac{\max a_{ij}}{\min a_{ij}}\bigr)$-competitive 
ratio where $a_{ij}$'s are entries in the covering matrix. 
In the approaches, it is not clear how to design competitive algorithms 
for \emph{non-convex} functions. A distinguishing point of our approach is that it gives a framework to study 
non-convex cost functions.

\paragraph{Organization.} In Section \ref{sec:general}, we present the framework for the general problem
described in Section~\ref{sec:intro-gen}. The applications of this framework are in Appendix
\ref{appendix:apps-general}. In Section \ref{sec:general}, we give the framework for the 0-1 covering problems where 
some proof can be found in Appendix \ref{sec:proof-bound-x}. Technical lemmas, which will be used to determined 
smooth and local-smooth parameters, are put in Appendix \ref{appendix:technical}. 

\section{Primal-Dual General Framework}		\label{sec:general}
In this section, we consider the general problem described in Section~\ref{sec:intro-gen} and 
present a primal-dual framework for this problem.

\paragraph{Formulation.}
We consider the formulation for the resource cost minimization problem 
following the configuration LP construction in \cite{MakarychevSviridenko14:Solving-optimization}.
We say that $A$ is a \emph{configuration} associated to resource $e$ if $A$ is a subset of requests using $e$.
Let $x_{ij}$ be a variable indicating whether request $i$ selects strategy (action) $s_{ij} \in \mathcal{S}_{i}$.  
For configuration $A$ and resource $e$, let $z_{eA}$ be a variable such that $z_{eA} = 1$ if and only if 
for every request $i \in A$, $x_{ij} = 1$ for some strategy $s_{ij} \in \mathcal{S}_{i}$ such that 
$e \in s_{ij}$. In other words, $z_{eA} = 1$ iff the set of requests using $e$ is exactly $A$. 
We consider the following formulation and the dual of its relaxation.

\begin{minipage}[t]{0.45\textwidth}
\begin{align*}
\min  \sum_{e,A} &f_{e}(A) z_{e,A} \\
\sum_{j:s_{ij} \in \mathcal{S}_{i}} x_{ij} &= 1 & &  \forall i \\
\sum_{A: i \in A} z_{eA}  &= \sum_{j: e \in s_{ij}} x_{ij}	& & \forall i, e \\
\sum_{A} z_{eA} &= 1 & & \forall e \\
x_{ij}, z_{eA} &\in \{0,1\} & & \forall i,j,e,A \\
\end{align*}
\end{minipage}
\quad
\begin{minipage}[t]{0.5\textwidth}
\begin{align*}
\max \sum_{i} \alpha_{i} &+ \sum_{e} \gamma_{e} \\
\alpha_{i} &\leq \sum_{e: e \in s_{ij}} \beta_{ie}  & &  \forall i,j \\
\gamma_{e} + \sum_{i \in A} \beta_{ie} &\leq f_{e}(A)  & & \forall e,A \\
\end{align*}
\end{minipage}

In the primal, the first constraint guarantees that request $i$ selects some strategy $s_{ij} \in \mathcal{S}_{i}$. 
The second constraint ensures that if request $i$ selects strategy $s_{ij}$ that contains resource $e$ then
in the solution, the set of requests using $e$ must contain $i$.
The third constraint says that in the solution, there is always a configuration associated to resource $e$. 

\paragraph{Algorithm.} 
We first interpret intuitively the dual variables, dual constraints and derive useful observations for a
competitive algorithm. Variable $\alpha_{i}$ represents the increase of the total cost due to the arrival of request $i$.
Variable $\beta_{i,e}$ stands for the marginal cost on resource $e$ if request $i$ uses $e$.
By this interpretation, the first dual constraint clearly indicates the behaviour of an algorithm. 
That is, if a new request $i$ is released, select a strategy $s_{ij} \in \mathcal{S}_{i}$ that minimizes the marginal increase
of the total cost. Therefore, we deduce the following greedy algorithm.

Let $A^{*}_{e}$ be the set of current requests using resource $e$. Initially, $A^{*}_{e} \gets \emptyset$ for every $e$. 
At the arrival of request $i$, select strategy $s^{*}_{ij}$ that is an optimal solution of 
\begin{align}	\label{eq:alpha}
\min \sum_{e \in s_{ij}} \biggl[ f_{e}(A^{*}_{e} \cup i) - f_{e}(A^{*}_{e}) \biggr]  
\qquad
\textnormal{s.t.}
\qquad
s_{ij} \in \mathcal{S}_{i}.
\end{align}
Although computational complexity is not a main issue for online problems, we notice that in many applications,
the optimal solution for this mathematical program can be efficiently computed (for example when $f_{e}$'s are convex
and $\mathcal{S}_{i}$ can be represented succinctly in form of a polynomial-size polytope).

\paragraph{Dual variables.} 
Assume that all resource cost $f_{e}$ are $(\lambda,\mu)$-smooth for some fixed parameters $\lambda > 0$ and $\mu < 1$. 
We are now constructing a dual feasible solution. Define $\alpha_{i}$ as $1/\lambda$ times the optimal value of the mathematical program 
(\ref{eq:alpha}). Informally, $\alpha_{i}$ is proportional the increase of the total cost due to the arrival of request $i$.
For each resource $e$ and request $i$, define 
$$
\beta_{i,e} := \frac{1}{\lambda} \biggl[ f_{e}(A^{*}_{e,\prec i} \cup i) - f_{e}(A^{*}_{e,\prec i})  \biggr]
$$
where $A^{*}_{e, \prec i}$ is the set of requests using resource $e$ (due to the algorithm) prior to the arrival of $i$.
In other words, $\beta_{ij}$ equals $1/\lambda$ times the marginal cost of resource $e$ if $i$ uses $e$.
Finally, for every resource $e$ define dual variable 
$
\gamma_{e} := - \frac{\mu}{\lambda} f_{e}(A^{*}_{e})
$
where $A^{*}_{e}$ is the set of all requests using $e$ (at the end of the instance).

\begin{lemma}
The dual variables defined as above are feasible. 
\end{lemma}
\begin{proof}
The first dual constraint follows immediately the definitions of $\alpha_{i}, \beta_{i,e}$
and the decision of the algorithm. 
Specifically, the right-hand side of the constraint represents $1/\lambda$ times 
the increase cost if the request selects a strategy $s_{ij}$. This is larger than 
$1/\lambda$ times the minimum increase cost optimized over all strategies in $\mathcal{S}_{i}$, 
which is $\alpha_{i}$. 

We now show that the second constraint holds. Fix a resource $e$ and a configuration $A$.
The corresponding constraint reads
\begin{align*}
& - \frac{\mu}{\lambda} f_{e}(A^{*}_{e}) + \frac{1}{\lambda} \sum_{i \in A} \biggl[ f_{e}(A^{*}_{e,\prec i} \cup i) - f_{e}(A^{*}_{e,\prec i}) \biggr]
\leq f_{e}(A) \\
\Leftrightarrow \qquad &
 \sum_{i \in A} \biggl[ f_{e}(A^{*}_{e,\prec i} \cup i) - f_{e}(A^{*}_{e,\prec i}) \biggr]
\leq \lambda f_{e}(A) + \mu f_{e}(A^{*}_{e})
\end{align*}
This inequality is due to the definition of $(\lambda,\mu)$-smoothness for resource $e$.
Hence, the second dual constraint follows.
\end{proof}

\setcounter{theorem}{0}
\begin{theorem}	\label{thm:main}
Assume that all cost functions are $(\lambda,\mu)$-smooth.
Then, the algorithm is $\lambda/(1-\mu)$-competitive.
\end{theorem}
\begin{proof}
By the definitions of dual variables, the dual objective is 
\begin{align*}
\sum_{i} \alpha_{i} + \sum_{e} \gamma_{e} = 
\sum_{e} \frac{1}{\lambda} f_{e}(A^{*}_{e})  - \sum_{e} \frac{\mu}{\lambda}  f_{e}(A^{*}_{e})
= \frac{1-\mu}{\lambda} \sum_{e}  f_{e}(A^{*}_{e})
\end{align*}
Besides, the cost of the solution due to the algorithm is $\sum_{e}  f_{e}(A^{*}_{e})$.
Hence, the competitive ratio is at most $\lambda/(1-\mu)$.
\end{proof}

\paragraph{Applications.}	
Theorem \ref{thm:main} yields simple algorithm with 
\emph{optimal} competitive ratios for several problems as mentioned 
in the introduction. Among others, we give optimal algorithms for energy efficient scheduling problems
(in unrelated machine environment) and the facility location with client-dependent cost problem. 
Prior to our work, no competitive algorithm has been known for the problems.
These applications can be found in Appendix~\ref{appendix:apps-general}.
The proofs are now reduced to compute smooth parameters $\lambda,\mu$ that subsequently imply the competitive ratios. 
We mainly use the smooth inequalities in Lemma \ref{lem:smooth-convex}, developed in \cite{CohenDurr12:Smooth-inequalities}, 
to derive the explicit competitive bounds in case of non-negative polynomial cost functions.

\section{Primal-Dual Framework for $0-1$ Covering Problems}		\label{sec:covering}
Consider the following integer optimization problem. 
Let $\mathcal{E}$ be a set of $n$ resources
and let $f: \{0,1\}^{n} \rightarrow \mathbb{R}^{+}$ be a monotone cost function.
Let $x_{e} \in \{0,1\}$ be a variable indicating whether resource $e$ is selected. 
The problem is to minimize $f(\vect{x})$ subject to covering constraints 
$\sum_{e} b_{i,e} x_{e} \geq 1$ for every constraint $i$ and $x_{e} \in \{0,1\}$
for every $e$. In the online setting, the constraints are revealed one-by-one and at any 
step, one needs to maintain a feasible integer solution $\vect{x}$.

We consider the \emph{multilinear extension} of function $f$ defined as follows.
Given $f$, define its multilinear extension $F: [0,1]^{n} \rightarrow \mathbb{R}^{+}$ by
$F(\vect{x}) := \sum_{S} f(\vect{1}_{S}) \prod_{e \in S} x_{e} \prod_{e \notin S} (1 - x_{e})$
where $\vect{1}_{S}$ is the characteristic vector of $S$ (i.e., the $e^{\textnormal{th}}$-component of 
$\vect{1}_{S}$ equals $1$ if $e \in S$ and equals 0 otherwise). Note that $F(\vect{1}_{S}) = f(\vect{1}_{S})$. 
An alternative way to define $F$ is to set $F(\vect{x}) = \mathbb{E} \bigl[ f(\vect{1}_{T})\bigr]$ where $T$ is a random set 
such that a resource $e$ appears in $T$ with probability $x_{e}$. 

In this section we will present an online algorithm that outputs a competitive fractional solution for function 
$F$ subject to the same set of constraints. The rounding schemes depend on specific problems. For example, one can benefit from 
techniques from \cite{AzarBuchbinder16:Online-Algorithms} in which rounding schemes have been given for several problems. 

\subsection{Algorithm for Fractional Covering}
We recall that a differentiable function $F: [0,1]^{n} \rightarrow \mathbb{R}^{+}$ is $(\lambda,\mu)$-locally-smooth
if for any subset $S$ of resources, the following inequality holds:
$$
\sum_{e \in S} \nabla_{e} F(\vect{x}) \leq \lambda F\bigl( \vect{1}_{S} \bigr) + \mu F\bigl( \vect{x} \bigr)
$$

\paragraph{Formulation.}
We say that $S \subset \mathcal{E}$ is a \emph{configuration} if $\vect{1}_{S}$ corresponds to a feasible solution.
Let $x_{e}$ be a variable indicating whether the resource $e$ is used.  
For configuration $S$, let $z_{S}$ be a variable such that $z_{S} = 1$ if and only if 
$x_{e} = 1$ for every resource $e \in S$, and $x_{e} = 0$ for $e \notin S$. 
In other words, $z_{S} = 1$ iff $\vect{1}_{S}$ is the selected solution of the problem.
For any subset $A \subset \mathcal{E}$, define $c_{i,A} = \max\{1 - \sum_{e' \in A} b_{i,e'}; 0\}$ 
and $b_{i,e,A} := \min\{b_{i,e}; c_{i,A}\}$.
We consider the following formulation and the dual of its relaxation.

\begin{minipage}[t]{0.45\textwidth}
\begin{align*}
\min  \sum_{S} &f(\vect{1}_{S}) z_{S} \\
\sum_{e \notin A} b_{i,e,A} \cdot x_{e} &\geq c_{i,A} & &  \forall i,A \subset \mathcal{E} \\
\sum_{S: e \in S} z_{S}  &= x_{e}	& & \forall e \\
\sum_{S} z_{S} &= 1 & & \\
x_{e}, z_{S} &\in \{0,1\} & & \forall e,S\\
\end{align*}
\end{minipage}
\quad
\begin{minipage}[t]{0.5\textwidth}
\begin{align*}
\max \sum_{i, A} c_{i,A} \alpha_{i,A} &+ \gamma \\
\sum_{i} \sum_{A: e \notin A} b_{i,e,A} \cdot \alpha_{i,A} &\leq \beta_{e}  & &  \forall e \\
\gamma + \sum_{e \in S} \beta_{e} &\leq f(\vect{1}_{S})  & & \forall S \\
\alpha_{i} &\geq 0 & & \forall i 
\end{align*}
\end{minipage}

In the primal, the first constraints are knapsack-constraints corresponding to the given polytope. 
The second constraint ensures that if a resource $e$ is chosen then
the selected solution must contain $e$.
The third constraint says that one solution (configuration) must be selected. 

\paragraph{Algorithm.} 
Assume that function $F(\cdot)$ is $\bigl(\lambda, \frac{\mu}{8\ln(1+2d^{2})}\bigr)$-locally smooth.
Let $d$ be the maximal number of positive entries in a row, i.e., $d = \max_{i} |\{b_{ie}: b_{ie} > 0\}|$.  
Consider the following Algorithm~\ref{algo:covering} which follows the scheme in \cite{AzarBuchbinder16:Online-Algorithms}
with some more subtle steps.

\begin{algorithm}[ht]
\begin{algorithmic}[1]  
\STATE Initially, set $A^{*} \gets \emptyset$. Intuitively, $A^{*}$ consists of all resources $e$ such that $x_{e} = 1$.
\STATE Upon the arrival of primal constraint $\sum_{e} b_{k,e} x_{e} \geq 1$ 
 and the new corresponding dual variable $\alpha_{k}$. 
\WHILE[\texttt{Increase primal, dual variables}]{$\sum_{e} b_{k,e,A^{*}} x_{e} < c_{k,A^{*}}$ \textbf{simultaneously}}  
	\STATE Increase $\tau$ with rate 1 and increase $\alpha_{k,A^{*}}$ at rate $\frac{1}{c_{k,A^{*}}} \cdot \frac{1}{\lambda \cdot \ln(1+2d^{2})}$.
		(Note that $c_{k,A^{*}} > 0$ by the condition of the while loop.)
	\FOR{$e \notin A^{*}$ such that $b_{k,e,A^{*}} > 0$ 
						}			
			\STATE Increase $x_{e}$ according to the following function
				\begin{align*}
					\frac{\partial x_{e}}{\partial \tau}	\gets \frac{b_{k,e,A^{*}} \cdot x_{e} + 1/d}{\nabla_{e} F(\vect{x})}
				\end{align*}
	\ENDFOR
	\STATE \textbf{if} $x_{e} = 1$ \textbf{then} update $A^{*} \gets A^{*} \cup \{e\}$.
	\FOR [\texttt{Decrease dual variables}]{$e \notin A^{*}$ such that 
			$\sum_{i=1}^{k} \sum_{A: e \notin A} b_{i,e,A} \cdot \alpha_{i,A} \geq \frac{1}{\lambda} \nabla_{e} F(\vect{x})$}
		\STATE Let $m^{*}_{e} \gets \arg \max \{b_{i,e,A^{*}}: \alpha_{i} > 0, 1 \leq i \leq k\}$.
		\STATE Increase $\alpha_{m^{*}_{e},A^{*}}$ continuously with rate 
			$- \frac{1}{c_{k,A^{*}}} \cdot \frac{b_{k,e,A^{*}}}{b_{m^{*}_{e},e,A^{*}}} \cdot\frac{1}{\lambda \cdot \ln(1+2d^{2})}$.
	\ENDFOR
\ENDWHILE
\end{algorithmic}
\caption{Algorithm for Covering Constraints.}
\label{algo:covering}
\end{algorithm}

\paragraph{Dual variables.} 
Variables $\alpha_{i}$ are constructed in the algorithm. Let $\vect{x}$ be the current solution of the algorithm.
Define $\beta_{e} = \frac{1}{\lambda} \nabla_{e} F(\vect{x})$ and $\gamma = -\frac{\mu}{8\lambda \cdot \ln(1+2d^{2})} F(\vect{x})$.

The following lemma gives a lower bound on $x$-variables where the proof is given in Appendix~\ref{sec:proof-bound-x}.

\begin{lemma}		\label{lem:bound-x}
Let $e$ be an arbitrary resource.  
During the execution of the algorithm, it always holds that 
$$
x_{e}	\geq \sum_{A: e \notin A} \frac{1}{\max_{i} b_{i,e,A} \cdot d} 
		\left[ \exp\biggl( \frac{\lambda \cdot \ln(1+2d^{2})}{\nabla_{e} F(\vect{x})} 
				\cdot \sum_{i} b_{i,e,A} \cdot \alpha_{i,A} \biggr) - 1 \right]
$$
\end{lemma}

\begin{lemma}
The dual variables defined as above are feasible. 
\end{lemma}
\begin{proof}
As long as a primal covering constraint is unsatisfied, the $x$-variables are always increased. Therefore, at the end of 
a iteration, the primal constraint is satisfied. Consider the first dual constraint. The algorithm always maintains that
$\sum_{i} \sum_{A: e \notin A} b_{i,e,A} \alpha_{i,A} \leq \frac{1}{\lambda} \nabla_{e} F(\vect{x})$ (strict inequality happens only if $x_{e} = 1$).
Whenever this inequality is violated then by the algorithm, some $\alpha$-variables are decreased in such a way that
the increasing rate of $\sum_{i} \sum_{A: e \notin A} b_{i,e,A} \alpha_{i,A}$ is at most 0.
Hence, by the definition of $\beta$-variables, the first dual constraint holds.

By definitions of dual variables and rearranging terms, the second dual constraint reads
\begin{align*}
	\frac{1}{\lambda} \sum_{e \in S} \nabla_{e} F(\vect{x}) \leq F(\vect{1}_{S}) + \frac{\mu}{8\lambda \cdot \ln(1+2d^{2})} F(\vect{x})
\end{align*}
This inequality is exactly the $\bigl(\lambda, \frac{\mu}{8\ln(1+2d^{2})}\bigr)$-local smoothness.
\end{proof}

We are now ready to prove the main theorem.

\begin{theorem}	\label{thm:main-fractional}
Assume that the cost function is $\bigl(\lambda, \frac{\mu}{\ln(1+2d^{2})}\bigr)$-locally-smooth.
Then, the algorithm is $O\bigl( \frac{\lambda}{1 - \mu}  \cdot \ln d \bigr)$-competitive.
\end{theorem}
\begin{proof}
We will bound the increases of the cost and the dual objective at any time $\tau$ in the execution of 
the algorithm. Let $A^{*}$ be the current set of resources $e$ such that $x_{e} = 1$.
The derivative of the objective with respect to $\tau$ is:
\begin{align*}
\sum_{e} \nabla_{e} F(\vect{x}) \cdot \frac{\partial x_{e}}{ \partial \tau}
&= \sum_{e: b_{k,e,A^{*}} > 0 \atop x_{e} < 1} \nabla_{e} F(\vect{x}) \cdot \frac{b_{k,e,A^{*}} \cdot x_{e} + 1/d}{\nabla_{e} F(\vect{x})} 
\leq  \sum_{e: b_{k,e,A^{*}} > 0} \biggl( b_{k,e,A^{*}} \cdot x_{e} + \frac{1}{d} \biggr) \leq 2
\end{align*}

For a time $\tau$, let $U(\tau)$ be the set of resources $e$ such that
$\sum_{i} \sum_{A: e \notin A} b_{i,e,A} \alpha_{i,A} = \frac{1}{\lambda} \nabla_{e} F(\vect{x})$ and $b_{k,e,A^{*}} > 0$.
Note that $|U(\tau)| \leq d$ by definition of $d$.
As long as $\sum_{e} b_{k,e} x_{e} < 1$ (so $\sum_{e \notin A^{*}} b_{k,e,A^{*}} x_{e} < c_{i,A^{*}} < 1$), 
for every $e \in U(\tau)$, 
by Lemma~\ref{lem:bound-x}, we have
\begin{align*}
\frac{1}{b_{k,e,A^{*}}} > x_{e} \geq \frac{1}{\max_{i} b_{i,e,A^{*}} \cdot d} 
				\left[ \exp\biggl(\ln(1+2d^{2}) \biggr) - 1 \right]
\end{align*}
Therefore, $\frac{b_{k,e,A^{*}}}{\max_{i} b_{i,e,A^{*}}} \leq \frac{1}{2d}$. 

We are now bounding the increase of the dual at time $\tau$.
The derivative of the dual with respect to $\tau$ is:
\begin{align*}
\frac{\partial D}{\partial \tau} 
&=  \sum_{i} \sum_{A} c_{i,A} \cdot \frac{\partial \alpha_{i,A}}{\partial \tau} + \frac{\partial \gamma}{\partial \tau}  
	= \sum_{i} c_{i,A^{*}} \cdot \frac{\partial \alpha_{i,A^{*}}}{\partial \tau} + \frac{\partial \gamma}{\partial \tau} \\
&= \frac{1}{\lambda \cdot \ln(1 + 2d^{2})} \biggl(1 - \sum_{e \in U(\tau)} \frac{b_{k,e,A^{*}}}{b_{m^{*}_{e},e,A^{*}}} \biggr) 
	- \frac{\mu}{8\lambda \cdot \ln(1 + 2d^{2})} \sum_{e} \nabla_{e} F(\vect{x}) \cdot \frac{\partial x_{e}}{\partial \tau} \\
&\geq \frac{1}{\lambda \cdot \ln(1 + 2d^{2})} \biggl(1 - \sum_{e \in U(\tau)} \frac{1}{2d} \biggr) 
	- \frac{\mu}{8 \lambda \cdot \ln(1 + 2d^{2})} \sum_{e \notin A^{*}: b_{i,e,A^{*}} > 0, x_{e} < 1} \biggl( b_{i,e,A^{*}}x_{e} + \frac{1}{d} \biggr) \\
&\geq \frac{1 - \mu}{4 \lambda \cdot \ln(1 + 2d^{2})}
\end{align*}
The third equality holds since $\alpha_{k,A^{*}}$ is increased and other $\alpha$-variables in 
$U(\tau)$ are decreased. The first inequality follows the fact that $\frac{b_{k,e,A^{*}}}{\max_{i} b_{i,e,A^{*}}} \leq \frac{1}{2d}$. 
The last inequality holds since 
$$\sum_{e \notin A^{*}: b_{i,e,A^{*}} > 0, x_{e} < 1} \bigl( b_{i,e,A^{*}}x_{e} + \frac{1}{d} \bigr) \leq 2$$
Hence, the competitive ratio is $O\bigl( \frac{\lambda}{1 - \mu} \cdot \ln d \bigr)$.
\end{proof}

\subsection{Applications}
In this section, we consider the applications of Theorem \ref{thm:main-fractional} for classes of cost functions 
which have been extensively studied in optimization such as polynomials with non-negative coefficients, 
$\ell_{k}$-norms and submodular functions. We are interested in deriving fractional solutions\footnote{Rounding schemes to obtaining integral solution for concrete problems are problem-specific and are not considered in this section. Several rounding techniques have been shown 
for different problems, for example in \cite{AzarBuchbinder16:Online-Algorithms} 
for polynomials with non-negative coefficients, or using online contention resolution schemes
for submodular functions \cite{FeldmanSvensson16:Online-contention}.} 
with performance guarantee. 
We show that Algorithm \ref{algo:covering} yields competitive fractional 
solutions for the classes of functions mentioned above and also for some natural classes of non-convex functions. 

We first take a closer look into the definition of local smoothness.
Let $F$ be a multilinear extension of a set function $f$. By definition of multilinear extension, 
$F(\vect{x}) = \mathbb{E} \bigl[ f(\vect{1}_{T})\bigr]$ where $T$ is a random set 
such that a resource $e$ appears in $T$ with probability $x_{e}$. Moreover, since $F$ is linear in $x_{i}$, we have
\begin{align*}
\frac{\partial F}{\partial x_{e}}(\vect{x}) 
&= F(x_{1}, \ldots, x_{e-1}, 1, x_{e+1}, \ldots, x_{n}) - F(x_{1}, \ldots, x_{e-1}, 0, x_{e+1}, \ldots, x_{n}) \\
&= \mathbb{E} \biggl[ f\bigl(\vect{1}_{R \cup \{e\}}\bigr) - f\bigl(\vect{1}_{R}\bigr) \biggr]
\end{align*} 
where $R$ is a random subset of resources $N \setminus \{e\}$ such that $x_{e'}$ is included with probability $x_{e'}$.
Therefore, in order to prove that $F$ is $(\lambda,\mu)$-locally-smooth, it is equivalent to show that
\begin{equation}	\label{eq:local-smooth-equiv}
\sum_{e \in S} \mathbb{E} \biggl[ f\bigl(\vect{1}_{R \cup \{e\}}\bigr) - f\bigl(\vect{1}_{R}\bigr) \biggr] 
\leq \lambda f\bigl( \vect{1}_{S} \bigr) + \mu \mathbb{E} \biggl[ f\bigl(\vect{1}_{R}\bigr) \biggr] 
\end{equation}

\paragraph{Polynomials with non-negative coefficients.}
Let $g: \mathbb{R} \rightarrow \mathbb{R}$ be a polynomial with non-negative coefficients and 
the cost function $f: \{0,1\}^{n} \rightarrow \mathbb{R}^{+}$ 
defined as $f(\vect{1}_{S}) = g\bigl( \sum_{e \in S} a_{e} \bigr)$ where $a_{e} \geq 0$ for every 
$e$. The following proposition shows that our algorithm yields the same competitive ratio as the one derived in
\cite{AzarBuchbinder16:Online-Algorithms} for this class of cost functions. 
This bound indeed is \emph{tight} \cite{AzarBuchbinder16:Online-Algorithms} (up to a constant factor).

\begin{proposition}[\cite{AzarBuchbinder16:Online-Algorithms}]	\label{prop:fractional-poly}
For any convex polynomial function $g$ of degree $k$, there exists an $O\bigl( (k \ln d)^{k} \bigr)$-competitive algorithm for 
the fractional covering problem.
\end{proposition}
\begin{proof}
We prove that Algorithm \ref{algo:covering} is $O\bigl( (k \ln d)^{k} \bigr)$-competitive for this class of cost functions. 
By Theorem \ref{thm:main-fractional}, it is sufficient to verify that $F$ is $((k \ln k)^{k-1}, \frac{k-1}{k \ln d})$-locally smooth.
We indeed prove a stronger inequality than (\ref{eq:local-smooth-equiv}), that is for any set $R$,
\begin{align*}	\label{eq:fractional-poly}
	\sum_{e \in S} \biggl[ f\bigl(\vect{1}_{R \cup \{e\}}\bigr) - f\bigl(\vect{1}_{R}\bigr) \biggr]  
		\leq O\bigl( (k \ln k)^{k-1} \bigr) \cdot f(\vect{1}_{S}) + \frac{k-1}{k \ln k} \cdot f\bigl(\vect{1}_{R}\bigr) 
\end{align*}
or equivalently, for any set $R$,
\begin{align*}
	\sum_{e \in S} \biggl[ g\biggl(a_{e} + \sum_{e' \in R} a_{e'}\biggr) - f\biggl(\sum_{e' \in R} a_{e'}\biggr) \biggr]  
		\leq O\bigl( (k \ln k)^{k-1} \bigr) \cdot g\biggl(\sum_{e \in S} a_{e}\biggr) + \frac{k-1}{k \ln k} \cdot g\biggl(\sum_{e' \in R} a_{e'}\biggr) 
\end{align*}
This inequality holds by Lemma~\ref{lem:smooth-convex}. 
Hence, the proposition follows.
\end{proof}

\paragraph{Norms functions.}
Let $g: \mathbb{R}^{n} \rightarrow \mathbb{R}$ be a function of weighted sum of $\ell_{k}$-norms, i.e.,
$g(\vect{x}) = \sum_{j=1}^{m} w_{j} \| x(S_{j}) \|_{k_{j}}$ where $S_{j}$ is a subset of resources and $w_{j} > 0$ for $1 \leq j \leq m$.
The cost function $f: \{0,1\}^{n} \rightarrow \mathbb{R}^{+}$ 
defined as $f(\vect{1}_{S}) = g\bigl( \vect{1}_{S} \bigr)$. 
This class of functions does not fall into the framework of \citet{AzarBuchbinder16:Online-Algorithms} 
since the corresponding gradient is not monotone. Very recently, \citet{NagarajanShen17:Online-Covering}
have overcome this difficulty and presented a nearly tight $O(\log d + \log \frac{\max a_{ij}}{\min a_{ij}})$-competitive algorithm where
$a_{ij}$'s are entries in the covering matrix. 
The lower bound is $\Omega(\log d)$ \cite{BuchbinderNaor09:Online-primal-dual}, that holds even for $\ell_{1}$-norm (linear costs).
Using Theorem \ref{thm:main-fractional}, we show that our algorithm 
yields \emph{tight} competitive ratio for this class of functions. 

\begin{proposition}
The algorithm is \emph{optimal} (up to a constant factor) for the class of weighted sum of $\ell_{k}$-norms 
with competitive ratio $O\bigl( \log d \bigr)$.
\end{proposition}
\begin{proof}
it is sufficient to verify that $F$ is $(1,0)$-locally smooth. Again, we
prove a stronger inequality than (\ref{eq:local-smooth-equiv}), that is for any index $1 \leq j \leq m$, for any set $R$,
\begin{align*}	
	\sum_{e \in S} w_{j} \left[ f\bigl(\vect{1}_{R \cup \{e\}}\bigr) - f\bigl(\vect{1}_{R}\bigr) \right]  
		\leq w_{j}f(\vect{1}_{S}) 
\quad
\Leftrightarrow
\quad
	\sum_{e \in S} w_{j} \bigl[ \| \vect{1}_{R \cup \{e\}} \|_{k_{j}} - \| \vect{1}_{R} \|_{k_{j}} \bigr]  
		\leq w_{j} \| \vect{1}_{S} \|_{k_{j}} 
\end{align*}
Note that, function $\bigl(\|\vect{b} + \vect{z}\|_{k_{j}} - \|\vect{b}\|_{k_{j}}\bigr)$ is convex (with respect to $\vect{z}$). 
Therefore, 
$$
	\sum_{e \in S}  \bigl[ \| \vect{1}_{R \cup \{e\}} \|_{k_{j}} - \| \vect{1}_{R} \|_{k_{j}} \bigr]  
	\leq  \| \vect{1}_{R \cup S} \|_{k_{j}} - \| \vect{1}_{R} \|_{k_{j}}  \leq \| \vect{1}_{S} \|_{k_{j}} 
$$
where the last inequality holds due to the triangle inequality of norms. The proposition follows.
\end{proof}

\paragraph{Beyond convex functions.}
Consider the following natural cost functions which represent 
more practical costs when serving clients as mentioned in the introduction
(the cost initially increases fast then becomes more stable before growing quickly again).
Let $g: \mathbb{R} \rightarrow \mathbb{R}$ be a non-convex function 
defined as $g(y) = y^{k}$ if $y \leq M_{1}$ or $y \geq M_{2}$
and $g(y) = g(M_{1})$ if $M_{1} \leq  y \leq M_{2}$ where $M_{1} < M_{2}$ are some constants. 
The cost function $f: \{0,1\}^{n} \rightarrow \mathbb{R}^{+}$ defined as 
$f(\vect{1}_{S}) = g\bigl( \sum_{e \in S} a_{e} \bigr)$ where $a_{e} \geq 0$ for every $e$. 
This function is $((k \ln k)^{k-1}, \frac{k-1}{k \ln d})$-locally smooth. Again, 
it sufficient to verify Inequality (\ref{eq:local-smooth-equiv}) and the proof is similar to the one in 
Proposition \ref{prop:fractional-poly} (or more specifically, Lemma \ref{lem:smooth-convex}) 
and note that the derivative of $g$ for $M_{1} < y < M_{2}$ equals 0. 

\begin{proposition}
The algorithm is $O\bigl( (k \ln d)^{k} \bigr)$-competitive  for minimizing the non-convex objective function defined above 
under covering constraints.
\end{proposition}

\paragraph{Submodular functions.}
Consider the class of submodular functions $f$, that is 
$f(\vect{1}_{S \cup \{e\}}) - f(\vect{1}_{S}) \geq f(\vect{1}_{T \cup \{e\}}) - f(\vect{1}_{T})$
for every $e$ and $S \subset T$ and $f(\vect{1}_{\emptyset}) = 0$. 
Submodular optimization has been extensively studying in optimization and machine learning. 
In the context of online algorithms, \citet{BuchbinderFeldman15:Online-submodular} have considered submodular optimization with preemption, where 
one can reject previously accepted elements, and have given constant competitive algorithms for unconstrained and knapsack-constraint 
problems. To the best of our knowledge, the problem of online submodular minimization under covering constraints have not been 
considered. 
  
An important concept in studying submodular functions is the \emph{curvature}. Given a submodular 
function $f$, the \emph{total curvature} $\kappa_{f}$ of $f$ is defined as \cite{ConfortiCornuejols84:Submodular-set-functions}
$$
\kappa_{f} = 1 - \min_{e} \frac{f(\vect{1}_{\mathcal{E}}) - f(\vect{1}_{\mathcal{E} \setminus \{e\}})}{f(\vect{1}_{\{e\}})}
$$
Intuitively, the total curvature mesures how far away $f$ is from being \emph{modular}. The concept of 
curvature has been used to determines both upper and lower bounds on the approximation ratios
for many submodular and learning problems \cite{ConfortiCornuejols84:Submodular-set-functions,GoemansHarvey09:Approximating-submodular,BalcanHarvey12:Learning-Submodular,Vondrak10:Submodularity-and-Curvature:,IyerJegelka13:Curvature-and-optimal,SviridenkoVondrak17:Optimal-approximation}.

In the following, we present a competitive algorithm for minimizing a monotone submodular 
function under covering constraints where the competitive ratio is characterized by the curvature of the function
(and also the sparsity $d$ of the covering matrix).
We first look at an useful property of the total curvature. 

\begin{lemma}		\label{lem:curvature}
For any set $S$, it always holds that
$$
f(\vect{1}_{S}) \geq (1-\kappa_{f}) \sum_{e \in S} f(\vect{1}_{\{e\}}).
$$
\end{lemma}
\begin{proof} 
Let $S = \{e_{1}, \ldots, e_{m}\}$ be an 
arbitrary subset of $\mathcal{E}$. Let $S_{i} = \{e_{1}, \ldots, e_{i}\}$ for $1 \leq i \leq m$ and $S_{0} = \emptyset$. 
We have
\begin{align*}
f(\vect{1}_{S}) 
&\geq  f(\vect{1}_{\mathcal{E}}) -  f(\vect{1}_{\mathcal{E} \setminus S})
= \sum_{i=0}^{m-1}  f(\vect{1}_{\mathcal{E} \setminus S_{i}}) - f(\vect{1}_{\mathcal{E} \setminus S_{i+1}}) 
\geq \sum_{i=1}^{m}  f(\vect{1}_{\mathcal{E}}) - f(\vect{1}_{\mathcal{E} \setminus \{e_{i}\}}) \\
&\geq (1 - \kappa_{f}) \sum_{i=1}^{m} f(\vect{1}_{e_{i}})
\end{align*}
where the first two inequalities is due to submodularity of $f$ and the last inequality follows by the definition of the curvature.
\end{proof}

\begin{proposition}
The algorithm is $O\bigl( \frac{\log d}{1 - \kappa_{f}} \bigr)$-competitive for minimizing monotone submodular function
under covering constraints.
\end{proposition}
\begin{proof}
It is sufficient to verify that $F$ is $\bigl(\frac{1}{1-\kappa_{f}},0\bigr)$-locally smooth. Indeed, the 
$\bigl(\frac{1}{1-\kappa_{f}},0\bigr)$-local smoothness holds due to the submodularity and Lemma~\ref{lem:curvature}:
for any subset $R$, 
\begin{align*}	
	\sum_{e \in S} \left[ f\bigl(\vect{1}_{R \cup \{e\}}\bigr) - f\bigl(\vect{1}_{R}\bigr) \right]  
		\leq \sum_{e \in S} \left[ f\bigl(\vect{1}_{\{e\}}\bigr) \right]  
		\leq \frac{1}{1 -\kappa_{f}} \cdot f(\vect{1}_{S}) 
\end{align*}
Therefore, the proposition follows.
\end{proof}

\section{Conclusion}
In this paper, we have presented primal-dual approaches based on configuration 
linear programs to design competitive algorithms for problems with non-linear/non-convex objectives.
Non-convexity until now is considered as a barrier in optimization. 
We hope that our approach would contribute some elements toward the study of non-linear/non-convex problems.
Our work gives raise to several concrete questions related to the online optimization problem under covering constraints. 
The local-smoothness has provided tight bounds for classes of polynomials with non-negative coefficients 
and sum of weighted $\ell_{k}$-norms. So a question is to determine tight bounds for different classes of functions.
Moreover, is there connection between local-smoothness and total curvature in submodular optimization? Intuitively, 
both concepts measure how far way a function is from being modular.


\paragraph{Acknowledgement.} We thank Nikhil Bansal and Seeun William Umboh for useful discussions
that improve the presentation of the paper.

\newpage
\appendix
\section*{Appendix}

\section{Applications of Theorem~\ref{thm:main}}
\label{appendix:apps-general} 
 
\subsection{Minimum Power Survival Network Routing}

\paragraph{Problem.} In the problem, we are given a graph $G(V,E)$ and requests arrive online. 
The demand of a request $i$ is specified by a source $s_{i} \in V$, a sink $t_{i} \in V$, the load vector $p_{i,e}$
for every edge (link) $e \in E$ and an integer number $k_{i}$. At the arrival of request $i$, 
one needs to choose $k_{i}$ edge-disjoint paths connecting $s_{i}$ to $t_{i}$. 
Request $i$ increases the load $p_{i,e}$ for each edge $e$ used to satisfy its demand.   
The load $\ell_{e}$ of an edge $e$ is defined as the total load incurred by the requests using $e$. 
The \emph{power} cost of edge $e$ with load $\ell_{e}$ is $f_{e}(\ell_{e})$. 
The objective is to minimize the total power $\sum_{e} f_{e}(\ell_{e})$.
Typically $f_{e}(\ell_{e}) = c_{e} \ell_{e}^{\alpha_{e}}$ where $c_{e}$ and $\alpha_{e}$ are parameters depending on $e$.

This problems generalizes the \textsc{Minimum Power Routing} problem --- a variant in which $k_{i} = 1$ and $p_{i,e} = 1 ~\forall i,e$ ---
and the \textsc{Load Balancing} problem --- a variant in which $k_{i}=1$, all the sources (sinks) are the same $s_{i} = s_{i'} ~\forall i,i'$ 
($t_{i} = t_{i'} ~\forall i,i'$) and every $s_{i}-t_{i}$ path has length 2.
For the \textsc{Minimum Power Routing} in offline setting, \citet{AndrewsAnta12:Routing-for-Power} gave a polynomial-time poly-log-approximation algorithm. The result has been improved by
\citet{MakarychevSviridenko14:Solving-optimization} who gave an $B_{\alpha}$-approximation algorithm. 
In online setting, \citet{GuptaKrishnaswamy12:Online-Primal-Dual} presented an $\alpha^{\alpha}$-competitive online
algorithm. For the \textsc{Load Balancing} problem, the currently best-known approximation is $B_{\alpha}$
due to \cite{MakarychevSviridenko14:Solving-optimization} via their rounding technique based on decoupling inequality.
In online setting, it has been shown that the optimal competitive ratio for the \textsc{Load Balancing} 
problem is $\Theta(\alpha^{\alpha})$ \cite{Caragiannis08:Better-bounds}.

\paragraph{Contribution.} In the problem, the set of strategy $\mathcal{S}_{i}$ for each request $i$ is a solution consists of 
$k_{i}$ edge-disjoint paths connecting $s_{i}$ and $t_{i}$. Applying the general framework, we deduce the following greedy algorithm.

Let $\ell_{e}$ be the load of edge $e$.
Initially, set $\ell_{e} \gets 0$ for every edge $e$. 
At the arrival of request $i$, compute a strategy consisting of $k_{i}$ edge-disjoint paths 
from $s_{i}$ and $t_{i}$ such that the increase of the total cost is minimum. 
Select this strategy for request $i$ and update $\ell_{e}$.

We notice that computing the strategy for request $i$ can be done efficiently. Given the current loads $\ell_{e}$ on every 
edge $e$, create a graph $H$ consisting of the same vertices and edges as graph $G$. For each edge $e$ in graph $H$,
define the capacity to be 1 and the cost on $e$  to be $f_{e}(p_{i,e} + \ell_{e}) - f_{e}(\ell_{e})$. Then the computing 
$k_{i}$ edge-disjoint paths from $s_{i}$ and $t_{i}$ with the minimal marginal cost in $G$ is equivalent to solving 
a transportation problem in graph $H$.

\begin{proposition}
If the congestion costs of all edges are $(\lambda,\mu)$-smooth then
the algorithm is $\lambda/(1-\mu)$-competitive. In particular, if $f_{e}(z) = z^{\alpha_{e}}$
then the algorithm is $O(\alpha^{\alpha})$-competitive where $\alpha = \max_{e} \alpha_{e}$.
\end{proposition}
\begin{proof}
The proposition follows directly from Theorem~\ref{thm:main} and the particular case is derived 
additionally by Lemma~\ref{lem:smooth-convex}. 
\end{proof}

Note that one can generalizes the problem to capture more general or different types of connectivity demands
and the congestion costs are incurred from vertices instead of edges. The same results hold.

\subsection{Online Vector Scheduling}

\paragraph{Problem.} 
In the problem, there are $m$ unrelated machines and jobs arrive online. The load of a job $j$ in machine $i$ 
is specified by a vector $p_{ij} = \langle p_{ij}(k): 1 \leq k \leq d\rangle$ where $p_{ij}(k) \geq 0$ and $d$, 
a fixed parameter, is the dimension of the vector. 
At the arrival of a job $j$, vectors $p_{ij}$ for all $i$ are revealed and job $j$ must be assigned 
immediately to a machine. Given a job-machine assignment $\sigma$, the load in dimension $k$ of machine $i$
is defined as $\ell_{i,\sigma}(k) := \sum_{j: \sigma(j) = i} p_{ij}(k)$ for $1 \leq k \leq d$. The $L_{\alpha}$-norm for $\alpha \geq 1$ 
in dimension $k$ is $\|\Lambda_{\sigma}(k)\|_{\alpha} := \bigl(\sum_{i=1}^{m} \ell_{i,\sigma}(k)^{\alpha} \bigr)^{1/\alpha}$; and the 
$L_{\infty}$-norm (makespan norm) in dimension $k$ is $\|\Lambda_{\sigma}(k)\|_{\infty}  := \max_{i=1}^{m} \ell_{i,\sigma}(k)$.  
In the $L_{\alpha}$-norm, the objective is to find an online assignment $\sigma$ minimizing $\max_{k} \|\Lambda_{\sigma}(k)\|_{\alpha}$.
In the $L_{\infty}$-norm, the objective is to find an online assignment $\sigma$ minimizing $\max_{k} \|\Lambda_{\sigma}(k)\|_{\infty}$.
An algorithm is $r$-competitive for the $L_{\alpha}$-norm if it outputs an assignment $\sigma$ such that for any assignment $\sigma^{*}$,
it holds that $\max_{k} \|\Lambda_{\sigma}(k)\|_{\alpha} \leq r \cdot \max_{k} \|\Lambda_{\sigma^{*}}(k)\|_{\alpha}$.
Similarly for the $L_{\infty}$-norm objective.

The online vector scheduling is introduced by \citet{ChekuriKhanna04:On-multidimensional-packing}.
Recently, \citet{ImKell15:Tight-bounds} showed an optimal competitive algorithm for this problem. 
Their analysis is based on a carefully constructed potential function. In the following, we can also derive
an optimal algorithm for this problem based on our general framework. The analysis is much simpler
and follows directly Theorem~\ref{thm:main}.

\paragraph{Contribution.}
In the problem, the set of strategy $\mathcal{S}_{i}$ for each job $j$ is a machine $i$. 
Applying the general framework, we deduce the following greedy algorithm. 

In the $L_{\alpha}$-norm objective for $1 \leq \alpha \neq \infty$, consider function 
$C(\sigma) := \sum_{k = 1}^{d} \bigl( \sum_{i=1}^{m} \ell_{i,\sigma}(k)^{\alpha} \bigr)^{\frac{\alpha + \log d}{\alpha}}$
where $\sigma$ is a job-machine assignment of all jobs released so far. 
In the $L_{\infty}$-norm objective, consider function
$C(\sigma) := \sum_{k = 1}^{d} \sum_{i=1}^{m} \ell_{i}(k)^{\alpha + \log d}$ for $\alpha = \log(m)$. 

Initially, $\sigma$ is an empty assignment.  
At the arrival of $j$, assign $j$ to machine $i^{*}$ that minimizes the increase of 
$C(\sigma)$. Again, the assignment of a job $j$ can be efficiently computed. 

\begin{proposition}[\cite{ImKell15:Tight-bounds}]
For the $L_{\alpha}$-norm objective, the algorithm is $O(\max\{\alpha, \log d\})$-competitive.
For the $L_{\infty}$-norm objective, the algorithm is $O(\log d + \log m)$-competitive.
\end{proposition}
\begin{proof}
Let $\sigma^{*}$ is an optimal assignment for the $L_{\alpha}$-norm objective.
We have
\begin{align*}
\biggl( \sum_{i=1}^{m} \ell_{i,\sigma}(k)^{\alpha} \biggr)^{\frac{\alpha + \log d}{\alpha}}
&\leq \sum_{k=1}^{d} \biggl( \sum_{i=1}^{m} \ell_{i,\sigma}(k)^{\alpha} \biggr)^{\frac{\alpha + \log d}{\alpha}} \\
&\leq O\bigl( \alpha + \log d \bigr)^{\alpha + \log d} \cdot \sum_{k=1}^{d} \biggl( \sum_{i=1}^{m} \ell_{i,\sigma^{*}}(k)^{\alpha} \biggr)^{\frac{\alpha + \log d}{\alpha}} \\
&\leq O\bigl( \alpha + \log d \bigr)^{\alpha + \log d} \cdot d \cdot \max_{k=1}^{d} \biggl( \sum_{i=1}^{m} \ell_{i,\sigma^{*}}(k)^{\alpha} \biggr)^{\frac{\alpha + \log d}{\alpha}}
\end{align*} 
In these inequalities, we apply  Theorem~\ref{thm:main} 
and Lemma~\ref{lem:smooth-convex} (note that $C(\sigma)$ is a polynomial of degree $(\alpha + \log d)$).
Taking the $(\alpha + \log d)^{\textnormal{th}}$ root, the result for $L_{\alpha}$-norm objective follows.

For the $L_{\infty}$-norm, with $\alpha = \log m$ similarly we have
\begin{align*}
\max_{k=1}^{d} \max_{i=1}^{m} \ell_{i,\sigma}(k)^{\alpha + \log d}
&\leq \sum_{k = 1}^{d}  \sum_{i=1}^{m} \ell_{i,\sigma}(k)^{\alpha + \log d} \\
&\leq \bigl( \alpha +  \log d \bigr)^{\alpha + \log d} \cdot \sum_{k=1}^{d} \sum_{i=1}^{m} \ell_{i,\sigma^{*}}(k)^{\alpha + \log d} \\
&\leq \bigl( \alpha +  \log d \bigr)^{\alpha + \log d} \cdot d \cdot m \cdot \max_{k=1}^{d} \max_{i=1}^{m} \ell_{i,\sigma^{*}}(k)^{\alpha + \log d} 
\end{align*} 
Again, taking the $(\alpha + \log d)^{\textnormal{th}}$ root for $\alpha = \log m$, the proposition follows. 
\end{proof}

\subsection{Online Energy-Efficient Scheduling}

\paragraph{Problem.} 
Energy-efficient algorithms have received considerable attention and has been widely studied in scheduling.
One main direction is to design performant algorithms toward a more realistic setting --- online setting 
with multiple machine environment \cite{Albers10:Energy-efficient-algorithms}.
We consider an energy minimization problem in the online multiple machine setting.
In the problem, we are given $m$ unrelated machines and a set of jobs. Each job $j$ is specified by 
its released date $r_{j}$, deadline $d_{j}$ and processing volumes $p_{ij}$ if job $j$ is processed in 
machine $i$. We consider \emph{non-migration} schedules; that is, 
every job $j$ has to be assigned to exactly one machine and is fully processed in that machine 
during time interval $[r_{j},d_{j}]$. However, jobs can be executed \emph{preemptively}, meaning that 
a job can be interrupted during its execution and can be resumed later on. An algorithm can choose 
appropriate speed $s_{i}(t)$ for every machine $i$ at any time $t$ in order to complete all jobs. 
Every machine $i$ has a non-decreasing energy power function $P_{i}(s_{i}(t))$ depending on the speed $s_{i}(t)$.
Typically, $P_{i}(z)$ has form $z^{\alpha_{i}}$ for constant $\alpha_{i} \geq 1$ or in a more general context,
$P_{i}(z)$ is assumed to be convex. In the problem, we consider general non-decreasing 
continuous
functions $P_{i}$ \emph{without} convexity assumption. The objective is to minimize the total energy consumption while completing all 
jobs. In the \emph{online} setting, jobs arrive over time and the assignment and scheduling have to be done irrevocably.

In offline setting, for typical energy function 
$P(z) = z^{\alpha}$, the best known algorithms \cite{GreinerNonner14:The-bell-is-ringing,BampisKononov13:Energy-Efficient} 
have competitive ratio $O(B_{\alpha})$ where $B_{\alpha}$ is the Bell number. 
Prior to our work, the only known result for this online problem is in the single machine setting and the energy power function $P(z) = z^{\alpha}$.
Specifically, \citet{BansalKimbrel07:Speed-scaling} gave a $2 \bigl( \frac{\alpha}{\alpha - 1} \bigr)^{\alpha} e^{\alpha}$-competitive algorithm. 
In terms of lower bounds, \citet{BansalChan12:Improved-Bounds} showed that 
no deterministic algorithm has competitive ratio less than $e^{\alpha-1}/\alpha$ for single machine.
For unrelated machines, the lower bound $\Omega(\alpha^{\alpha})$ follows the construction of 
\citet{Caragiannis08:Better-bounds} for \textsc{Load Balacing} (with $L_{\alpha}$-norm)  
(by considering all jobs have the same span $[r_{j},d_{j}] = [0,1]$).
\citet{KlingPietrzyk15:Profitable-scheduling} gave a $O(\alpha^{\alpha})$-competitive algorithm 
in the multi-identical-processor setting in which job migration is allowed. 
Surprisingly, no competitive algorithm is known in the non-migratory multiple-machine environment, 
that is in contrast to the similar online problem with objective as the total energy plus flow-time \cite{AnandGarg12:Resource-augmentation}. 
The main difference here is that for the latter, one can make a tradeoff between energy and flow-time 
and derive a competitive algorithm whereas for the former, one has to deal directly with a non-linear objective   
and no LP with relatively small integrality gap was known. 
We notice that \citet{GuptaKrishnaswamy12:Online-Primal-Dual} gave also a primal-dual 
competitive algorithm for the single machine environment. 
However, their approach cannot be used for unrelated machines due to the large integrality gap 
of the formulation.


\paragraph{Contribution.}
In the problem, speed of a job can be an arbitrary (non-negative) real number. 
However, in order to employ tools from linear programming, we consider discrete setting
with a small loss in the competitive ratio.
Fix an arbitrary constant $\epsilon > 0$ and $\delta > 0$. Define the set of speeds 
$\mathcal{V} = \{\ell \cdot \epsilon: 0 \leq \ell \leq L\}$ for some sufficiently large $L$.
During a time interval $[t,t+\delta]$, a job can be executed at a speed in $\mathcal{V}$.  
As the energy cost functions are continuous, this assumption on the setting worsens the energy cost 
by at most a factor $(1+\tilde{\epsilon})$ for arbitrarily small $\tilde{\epsilon}$.
Given a job $j$, a set of feasible strategy $\mathcal{S}_{j}$ of $j$ is a feasible non-migratory execution
of a job $j$ on some machine. Specifically, a strategy of job $j$ can be described as the union over all machines $i$ 
of solutions determined by the following program:
$
	\sum_{t = r_{j}}^{d_{j}} \delta \cdot v_{ijt} \geq p_{ij}
	\textnormal{ s.t. }
	v_{ijt} \in \mathcal{V},
$
where in the sum we increment each time $t$ by $\delta$.
Applying the general framework, we derive the following algorithm. 

Let $u_{it}$ be the speed of machine $i$ at time $t$.
Initially, set $u_{it} \gets 0$ for every machine $i$ and time $t$. 
At the arrival of a job $j$, compute the minimum energy increase if job $j$ is assigned to machine $i$.
It is indeed an optimization problem 

\begin{equation}	\label{eq:local-opt}
\min \sum_{r_{j}}^{d_{j}} \delta \cdot \biggl[ P_{i}\bigl(u_{it} + v_{ijt}\bigr) - P_{i}\bigl(u_{it}\bigr) \biggr]
\quad 
\text{s.t}
\quad
\sum_{r_{j}}^{d_{j}} \delta \cdot v_{ijt} \geq p_{ij}, \quad v_{ijt} \in \mathcal{V}
\end{equation}

Observe that if $P_{i}$ is a convex function then it is a convex program and can be solved efficiently. 
In this case, using the KKT conditions, the optimal solution can be constructed as follows.  
We initiate a variable $v_{ijt}$ as $0$. While $\sum_{r_{j}}^{d_{j}} \delta \cdot v_{ijt} < p_{ij}$, i.e., the total volume
of job $j$ has not been completed, continue increasing $v_{ijt}$ 
at $\arg \min_{r_{j} \leq t \leq d_{j}} u_{it} + v_{ijt}$. Note that this is exactly algorithm OA in 
\cite{BansalKimbrel07:Speed-scaling} for a single machine.
Let $v^{*}_{i^{*}jt}$ be an optimal solution of the mathematical program (\ref{eq:local-opt}). 
Then, assign job $j$ to machine $i^{*} \in \arg\min_{i} \beta_{ij}$ and execute $j$ at time $t$
by speed $v^{*}_{i^{*}jt}$.

\begin{proposition}
If the energy cost functions are $(\lambda,\mu)$-smooth then
the algorithm is $(1 + \epsilon) \lambda/(1-\mu)$-competitive 
for arbitrarily small $\epsilon$. In particular, if $P_{i}(z) = z^{\alpha_{i}}$
then the algorithm is $(1+\epsilon) O(\alpha^{\alpha})$-competitive where $\alpha = \max_{i} \alpha_{i}$.
\end{proposition}
\begin{proof}
The proposition follows directly from Theorem~\ref{thm:main}.
In the particular case $P_{i}(z) = z^{\alpha_{i}}$, the functions are $(\lambda,\mu)$-smooth 
with $\mu = (\alpha-1)/\alpha$ and $\lambda = O(\alpha^{\alpha-1})$ by Lemma~\ref{lem:smooth-convex}. 
The competitive ratio of this case follows.
\end{proof}

\subsection{Online Prize Collecting Energy-Efficient Scheduling}

\paragraph{Problem.}
We consider the same setting as in the \textsc{Energy Minimization} problem. Additionally, each job $j$
has a penalty $\pi_{j}$. There is no penalty from job $j$ if it is completely executed during $[r_{j},d_{j}]$ 
in some machine $i$. Otherwise, if job $j$ is not completed 
(even most volume of job $j$ have been executed) then
the algorithm has to pay a penalty $\pi_{j}$. The objective is to minimize 
the total penalty of uncompleted jobs plus the energy cost. 

\paragraph{Contribution.}
The result does not follow immediately from Theorem~\ref{thm:main} but the approach is exactly the one in 
the general framework. 

By the same formulation as the previous section, assume that the set of speeds is finite and discrete.  
The set of feasible strategy $\mathcal{S}_{j}$ of a job $j$ is a feasible non-migratory execution
of a job $j$ on some machine, defined in the previous section. The sets $\mathcal{S}_{j}$'s are also finite and discrete. 
We say that $A$ is a \emph{configuration} of machine $i$ if it is a schedule of a subset of jobs in $i$. Specifically, a configuration 
$A$ consists of tuples $(i,j,k)$ specifying that a job $j$ is assigned to machine $i$ and is executed according to 
strategy $s_{jk} \in \mathcal{S}_{j}$. 

We are now formulating a configuration LP for the problem.
Let $x_{ijk}$ be a variable indicating whether job $j$ is processed in machine $i$ according to strategy $s_{jk} \in \mathcal{S}_{j}$. 
For configuration $A$ and machine $i$, let $z_{iA}$ be a variable such that $z_{iA} = 1$ if and only if 
$x_{ijk} = 1$ for every $(i,j,k) \in A$. In other words, $z_{iA} = 1$ iff $A$ is the solution of the problem restricted on machine $i$.
Let $c_{i,A}$ be the energy cost of configuration $A$ in machine $i$.
We consider the following formulation.

\begin{align*}
\min  \sum_{i,A} c_{i,A} z_{iA} &+ \sum_{j}  \biggl( 1 - \sum_{i,k} x_{ijk} \biggr) \pi_{j} \\
\sum_{i,k} x_{ijk} &\leq 1 & &  \forall j \\
\sum_{A: (i,j,k) \in A} z_{iA}  &= x_{ijk}	& & \forall i,j,k \\
\sum_{A} z_{iA} &= 1 & & \forall i \\
x_{ijk}, z_{iA} &\in \{0,1\} & & \forall i,j,A \\
\end{align*}

The first constraint guarantees that a job $j$ can be assigned to at most one machine $i$
and is executed according to at most one feasible strategy. 
The second constraint ensures that if job $j$ is assigned to machine $i$ and is executed according 
to strategy $s_{jk} \in \mathcal{S}_{j}$ then the configuration corresponding to the solution restricted
on machine $i$ must contain $(i,j,k)$. The third constraint says that there is always a configuration associated to 
machine $i$ for every $i$. The dual of the relaxation reads

\begin{align*}
\max \sum_{j} (\pi_{j} - \alpha_{j}) &+ \sum_{i} \gamma_{i} \\
\alpha_{j} + \beta_{ijk} &\geq  \pi_{j}  & &  \forall i,j,k \\
\gamma_{i} + \sum_{(i,j,k) \in A} \beta_{ijk} &\leq c_{iA}  & & \forall i,A \\
\alpha_{j} &\geq 0 & & \forall j 
\end{align*}

\paragraph{Greedy Algorithm.}
Assume that all energy power functions are $(\lambda,\mu)$-smooth. Fix $\lambda$ and $\mu$.
At the arrival of job $j$, compute the minimum energy increase if $j$ is assigned to some machine $i$. 
If the minimum energy increase is larger than $\lambda \cdot \pi_{j}$ then reject the job. Otherwise,
assign and execute $j$ such that the energy increase is minimum.  

\begin{proposition}
Assume that all energy power functions are $(\lambda,\mu)$-smooth. Then the algorithm
is $\lambda/(1-\mu)$-competitive.
\end{proposition}
\begin{proof}
We define the dual variables similarly as in the general framework. 
Let $A^{*}_{i,\prec j}$ be configuration of machine $i$ (due to the algorithm) before the arrival of job $j$. 
(Initially, $A^{*}_{i,\prec 1} \gets \emptyset$ for every machine $i$.)
For each machine $i$ and a strategy $s_{jk} \in \mathcal{S}_{j}$ such that $s_{jk}$ is a schedule of $j$ in machine $i$, 
define 
$$
\beta_{ijk} = \frac{1}{\lambda} \biggl[ c_{i}(A^{*}_{i,\prec j} \cup s_{jk}) - c_{i}(A^{*}_{i,\prec j})  \biggr].
$$
If $s_{jk}$ is not a schedule of $j$ in machine $i$ then define $\beta_{ijk} = \infty$.
Moreover, define 
\begin{align*}
\alpha_{j} = \max \bigl \{\pi_{j} - \min_{i,k} \beta_{ijk}, 0  \bigr\}
\qquad
\textnormal{and}
\qquad
\gamma_{i} = \frac{\mu}{\lambda} c_{i}(A^{*}_{i})
\end{align*}
where $A^{*}_{i}$ is the configuration of machine $i$ at the end of the instance (when all jobs have been released).  

The variables constitute a dual feasible solution. The first dual constraint follows the definition of 
$\alpha_{j}$. The second dual constraint follows the definition of $(\lambda,\mu)$-smoothness. 
Note that that for any configuration $A$ of a machine $i$ (a feasible schedule in machine $i$), 
if $(i,j,k) \in A$ then by definition of dual variables, $\beta_{ijk} \neq \infty$. 

We are now bounding the dual. The algorithm has the property immediate-reject. It means that if
the algorithm accepts a job then the job will be completed; and otherwise, the job is rejected at its arrival. 
By the algorithm, $\alpha_{j} = 0$ for every rejected job $j$. Besides, if job $j$ is accepted then 
$\pi_{j} - \alpha_{j} = \beta_{ijk}$ where $i$ is the machine to which job $j$ is assigned and job $j$
is executed according to strategy $s_{jk}$. Therefore, by the definition of dual variables, 
$\sum_{j}(\pi_{j} - \alpha_{j})$, where the sum is taken over accepted jobs $j$, equals $1/\lambda$ times
the total energy consumption. Recall that the total energy consumption of the algorithm is $\sum_{i} c_{i}(A^{*}_{i})$. The dual objective is
\begin{align*}
\sum_{j}(\pi_{j} - \alpha_{j}) + \sum_{i} \gamma_{i}
= \sum_{j: j \textnormal{ rejected}} \pi_{j} + \frac{1}{\lambda} \sum_{i} c_{i}(A^{*}_{i}) -  \frac{\mu}{\lambda} \sum_{i} c_{i}(A^{*}_{i})
\end{align*}
Moreover, the primal is equal to the total penalty of rejected jobs plus $\sum_{i} c_{i}(A^{*}_{i})$. Therefore, the ratio between 
primal and dual is at most $\lambda/(1-\mu)$.
\end{proof}

\subsection{Facility Location with Client-Dependent Facility Cost}
\paragraph{Non-Convex Facility Location.}
In the problem, we are given a metric space $(M,d)$ is given and clients arrive online. 
Let $N$ and $n$ be the set and the number of clients, respectively. 
A location $i \in M$ is characterized by a fixed opening cost $a_{i}$ 
and an arbitrary non-decreasing serving cost function $f_{i}: 2^{N} \rightarrow \mathbb{R}^{+}$.
If a subset $S$ of clients is served by a facility at location $i$ then 
the facility cost at this location is $a_{i} + f_{i}(S)$. At the arrival of a client, an algorithm 
need to assign the client to some facility. The goal is to minimize the total cost, which 
is the total distance from clients to their facilities plus the total facility cost. 

Facility Location is one of the most widely studied problems. In the classic version, the facility cost consists only of the opening 
cost. There is a large literature in the offline setting. In online setting, \citet{Meyerson01:Online-facility} gave 
a randomized $O(\frac{\log n}{\log \log n})$-competitive algorithm. This competitive ratio matches to the randomized lower bound
due to \citet{Fotakis08:On-the-competitive-ratio}. For deterministic algorithms, \citet{Fotakis07:A-primal-dual-algorithm} 
first presented a primal-dual $O(\log n)$-competitive algorithm and
subsequently improved to the optimal $O(\frac{\log n}{\log \log n})$-competitive algorithm \cite{Fotakis08:On-the-competitive-ratio}.
The online capacitated facility location in which function $f_{i}(S) = 0$ if $|S| \leq u_{i}$ for some capacity $u_{i}$
and $f_{i}(S) = \infty$ has been studied in \cite{AzarBhaskar13:Online-mixed}. Using a primal-dual framework for mixed 
packing and covering constraints, the authors derived a $O(\log m \log mn)$-competitive algorithm.  

\paragraph{Contribution.} We derive a competitive algorithm by combining the primal-dual algorithm due to
\citet{Fotakis07:A-primal-dual-algorithm} for the online (classic) facility location and our primal-dual framework 
for non-convex functions. 

Let $x_{ij}$ and $y_{i}$ be variables indicating whether client $j$ is assigned to facility $i$ and 
whether facility $i$ is open, respectively.
For subset $S \subset N$, let $z_{i,S}$ be a variable such that $z_{i,S} = 1$ if and only if 
$x_{ij} = 1$ for every client $j \in S$, and $x_{e} = 0$ for $j \notin S$. 
We consider the following formulation and the dual of its relaxation.

\begin{minipage}[t]{0.4\textwidth}
\begin{align*}
\min  \sum_{i} a_{i}y_{i} + \sum_{i,j} d_{ij}&x_{ij} + \sum_{i,S} f_{i}(S) z_{i,S} \\
\sum_{i} x_{ij} &\geq 1 \qquad  \forall j \\
y_{i} &\geq x_{ij} \qquad \forall i,j \\
\sum_{S: j \in S} z_{i,S}  &= x_{ij}	\qquad \forall i,j \\
\sum_{S} z_{i,S} &= 1 \qquad \forall i\\
x_{ij}, z_{i,S} &\in \{0,1\} \qquad \forall i,j,S\\
\end{align*}
\end{minipage}
\quad
\begin{minipage}[t]{0.5\textwidth}
\begin{align*}
\max \sum_{j} \alpha_{j} &+ \sum_{i} \theta_{i} \\
\alpha_{j} &\geq d_{ij} + \beta_{ij} + \gamma_{ij}   & &  \forall i,j \\
\sum_{j} \beta_{ij} &\leq a_{i} & & \forall i \\
\theta_{i} + \sum_{j \in S} \gamma_{ij} &\leq f_{i}(S)  & & \forall i,S \\
\alpha_{j}, \beta_{ij} &\geq 0 & & \forall i,j 
\end{align*}
\end{minipage}

\paragraph{Algorithm.} 
Assume that all serving cost $f_{i}$ are $(\lambda,\mu)$-smooth.
Intuitively, $\beta_{ij}$ and $\gamma_{ij}$ can be interpreted as the contributions of client $j$ to the opening cost and 
the serving cost at location $i$. At the arrival of client $j$, continuously increase $\alpha_{j}$. For any facility such that 
$\alpha_{j} = d_{ij}$, start increasing $\beta_{ij}$. If $\sum_{j'} \beta_{ij'} = a_{i}$ then (stop increasing $\beta_{ij}$) 
start increasing $\gamma_{ij}$ until $\frac{\mu}{\lambda}\bigl[ f_{i}(S \cup j) - f_{i}(S) \bigr]$ where $S$ is the current set of clients assigned to $i$. 
Assign $j$ to the first facility $i$ such that $\gamma_{ij} = \frac{\mu}{\lambda}\bigl[ f_{i}(S \cup j) - f_{i}(S) \bigr]$ 
(open $i$ if it has not been opened). 

\begin{proposition}
Assume that all serving cost $f_{i}$ are $(\lambda,\mu)$-smooth. Then the algorithm is 
$O\bigl(\log n + \frac{\lambda}{1-\mu}\bigr)$-competitive.
\end{proposition}
\begin{proof}
We define dual variables similarly as in Theorem~\ref{thm:main}. The $\alpha$-variables,
$\beta$-variables and $\gamma$-variables are defined in the algorithm. 
Define $\theta_{i}$ equal $-1/\lambda$ times the (final) serving cost at facility $i$. 
Let $\pi(j)$ is the facility to which $j$ is assigned and $\pi(N)$ the set of all open facilities by the algorithm. 

The dual variables constitute a feasible solution. The first and second dual constraints are due to the algorithm.
Note that by the definition of $\gamma$-variables, it always holds that
$\gamma_{ij} \leq \frac{\mu}{\lambda}\bigl[ f_{i}(S \cup j) - f_{i}(S) \bigr]$ where $S$ is the set of clients assigned to $i$
before the arrival of $j$.
The last constraint follows the $(\lambda,\mu)$-smoothness of serving costs. We are now bounding the primal and the dual. 
We have 
\begin{align*}
	\sum_{i \in \pi(N)} a_{i} + \sum_{j} d_{\pi(j),j} &\leq O(\log n) \sum_{j} \biggl( \alpha_{j} - \gamma_{\pi(j),j} \biggr)\\
		\sum_{i} f_{i}\bigl(\pi^{-1}(i)\bigr) &\leq \frac{\lambda}{1-\mu} \biggl( \sum_{j} \gamma_{\pi(j),j} + \sum_{i} \theta_{i} \biggr)
\end{align*}
where the first inequality is due to \citet{Fotakis07:A-primal-dual-algorithm} and the second one follows the definition of dual variables.
The proposition follows.
\end{proof}

\section{Proof of Lemma~\ref{lem:bound-x}}	\label{sec:proof-bound-x} 

\setcounter{lemma}{1}

\begin{lemma}		
Let $e$ be an arbitrary resource.  
During the execution of the algorithm, it always holds that 
$$
x_{e}	\geq \sum_{A: e \notin A} \frac{1}{\max_{i} b_{i,e,A} \cdot d} 
		\left[ \exp\biggl( \frac{\lambda \cdot \ln(1+2d^{2})}{\nabla_{e} F(\vect{x})} 
				\cdot \sum_{i} b_{i,e,A} \cdot \alpha_{i,A} \biggr) - 1 \right]
$$
\end{lemma}
\begin{proof}
We prove the lemma by induction. At the beginning of the instance, while no request has been released, 
both sides of the lemma are 0. Assume that the lemma holds until the arrival of $k^{\textnormal{th}}$ request.
Consider a moment $\tau$ and let $A^{*}$ be the current set of resources $e$ such that $x_{e} = 1$. 
We first consider the case $x_{e} < 1$.
The derivative of the right hand side according to $\tau$ is 
\begin{align*}
&\sum_{i} \frac{\partial \alpha_{i,A^{*}}}{\partial \tau} \cdot 
	\frac{b_{i,e,A^{*}}}{\max_{i} b_{i,e,A^{*}} \cdot d} \cdot \frac{\lambda \cdot \ln(1+2d^{2})}{ \nabla_{e} F(\vect{x})} 
		\cdot \exp\biggl( \frac{\lambda \cdot \ln(1+2d^{2})}{ \nabla_{e} F(\vect{x})} \cdot \sum_{i} b_{i,e,A} \alpha_{i,A} \biggr) \\
&\leq  \frac{b_{k,e,A^{*}} \cdot x_{e} + 1/d}{\nabla_{e} F(\vect{x})}
= \frac{\partial x_{e}}{\partial \tau}
\end{align*}
where in the inequality, we use the induction hypothesis; $\frac{\partial \alpha_{k,A^{*}}}{\partial \tau} > 0$ 
and $\frac{\partial \alpha_{i,A^{*}}}{\partial \tau} \leq 0$ for $i \neq k$; and the increasing rate of 
$\alpha_{k,A^{*}}$ according to the algorithm.
So the rate in the left-hand side is always larger than that in the right-hand side. 
Moreover, at some steps in the algorithm, $\alpha$-variables might be decreased while the $x$-variables
are maintained monotone. Hence, the lemma inequality holds.

The remaining case is $x_{e} = 1$. In this case, by the algorithm, the set $A^{*}$ has been updated so that 
$e \in A^{*}$. The increasing rates of both sides of the lemma inequality are 0. Therefore, the lemma follows.
\end{proof}

\section{Technical Lemmas}  	\label{appendix:technical}

In this section, we show technical lemmas in order to determine smoothness parameters
for polynomials with non-negative coefficients. The following lemma has been proved in 
\cite{CohenDurr12:Smooth-inequalities}. We give it here for completeness.

\setcounter{lemma}{4}
\begin{lemma}[\cite{CohenDurr12:Smooth-inequalities}]
  \label{lem:smooth-simple}
Let $k$ be a positive integer.
Let $0< a(k) \leq 1$ be a function on $k$. Then, for any  $x, y > 0$, it holds that
$$
y(x+y)^{k} \leq \frac{k}{k+1}a(k) x^{k+1} + b(k) y^{k+1} 
$$
where $\alpha$ is some constant and
\begin{subnumcases}
 {\label{b} b(k) =} 
	 \Theta\left(\alpha^{k} \cdot \left(\frac{k}{\log ka(k)} \right)^{k-1}\right) & \text{if} 
		$\lim_{k\rightarrow\infty}(k-1)a(k) = \infty$,	\label{b1} \\
	 \Theta\left(\alpha^{k} \cdot k^{k-1}\right)	 & \text{if} $(k-1)a(k)$ \text{are bounded} $\forall k$, \label{b2}\\
	 \Theta\left(\alpha^{k} \cdot \frac{1}{ka(k)^{k}}\right) & \text{if}
		$\lim_{k\rightarrow\infty} (k-1)a(k) = 0$.	\label{b3} 
\end{subnumcases}
\end{lemma}
\begin{proof}
Let $f(z) := \frac{k}{k+1}a(k) z^{k+1} - (1 + z)^{k} + b(k)$. 
To show the claim, it is equivalent to prove that $f(z) \geq 0$ for all $z > 0$. 

We have $f'(z) = ka(k)z^{k} - k(1+z)^{k-1}$. We claim that the equation 
$f'(z)=0$ has an unique positive root $z_{0}$. Consider the equation 
$f'(z)=0$ for $z > 0$. It is equivalent to 
$$
\left(\frac{1}{z} + 1\right)^k \cdot \frac{1}{z} = a(k)
$$
The left-hand side is a strictly decreasing function and the limits 
when $z$ tends to $0$ and $\infty$ are $\infty$ and $0$, respectively.
As $a(k)$ is a positive constant, there exists an unique root 
$z_0 > 0$. 

Observe that function $f$ is decreasing in $(0,z_{0})$ and 
increasing in $(z_{0}, +\infty)$, so $f(z) \geq f(z_{0})$ for all $z > 0$. 
Hence, by choosing 
\begin{equation}	\label{eq:bk}
b(k) = \Big | \frac{k}{k+1}a(k) z_{0}^{k+1} - (1 + z_{0})^{k} \Big |
= (1+z_{0})^{k-1}\Big(1+\frac{z_{0}}{k+1}\Big)
\end{equation}
it follows that $f(z) \geq 0 ~\forall z>0$.

We study the positive root $z_{0}$ of equation 
\begin{equation}	\label{eq:smooth-simple-derivative}
a(k) z^{k} - (1+z)^{k-1} = 0
\end{equation}
Note that $f'(1) = k(a(k) - 2^{k-1}) < 0$ since $0< a(k) \leq 1$. 
Thus, $z_{0} > 1$.  For the sake of simplicity, we
define the function $g(k)$ such that $z_{0} = \frac{k-1}{g(k)}$ where
$0 < g(k) < k-1$.  Equation~(\ref{eq:smooth-simple-derivative}) is
equivalent to
\begin{equation*} \label{eq:majoration}
\left(1 + \frac{g(k)}{k-1}\right)^{k-1}g(k) = (k-1)a(k)
\end{equation*}
Note that $e^{w/2} < 1 + w < e ^{w}$ for $w \in (0,1)$.  For $w := \frac{g(k)}{k-1}$,
we obtain the following upper and lower bounds for the term  $(k-1)a(k)$:
 
\begin{equation} \label{eq:LambertW}
e^{g(k)/2}g(k) < (k-1)a(k) < e^{g(k)}g(k)
\end{equation}

Recall the definition of \emph{Lambert $W$ function}.  
For each $y \in \mathbb{R}^{+}$, $W(y)$ is 
defined to be  solution of  the equation $xe^{x}  = y$.  Note that, 
$xe^{x}$  is increasing with respect to $x$, hence $W(\cdot)$ is increasing.

By definition of the Lambert $W$ function and
Equation~(\ref{eq:LambertW}), we get that
\begin{equation} \label{eq:encadrementG}
W\left((k-1)a(k)\right) < g(k) < 2W\left(\frac{(k-1)a(k)}{2}\right)
\end{equation}


First, consider  the case where  $\lim_{k\rightarrow\infty}(k-1)a(k) =
\infty$.  The asymptotic sequence for $W(x)$ as $x \to +\infty$ is the
following:    $W(x)=\ln    x-\ln   \ln    x+\frac{\ln\ln    x}{\ln
  x}+O\left(\left(\frac{\ln\ln x}{\ln x}\right)^2\right)$.  So, for large enough $k$, $W((k-1)a(k)) = \Theta(\log((k-1)a(k)))$.  Since
$z_{0} =  \frac{k-1}{g(k)}$, from Equation~(\ref{eq:encadrementG}), we
get $z_{0}  = \Theta\left( \frac{k}{\log (ka(k))}  \right)$.  
Therefore, by (\ref{eq:bk}) we have 
$b(k) = \Theta\left(\alpha^{k}\cdot\left(\frac{k}{\log ka(k)}
  \right)^{k-1}\right)$ for some constant $\alpha$.

Second,  consider  the  case  where  $(k-1)a(k)$ is  bounded  by  some
constants.  So by~(\ref{eq:encadrementG}), we have $g(k) = \Theta(1)$.
Therefore $z_{0} = \Theta(k)$ which again implies 
$b(k)  =  \Theta\left(\alpha^{k}\cdot  k^{k-1}\right)$  for  some
constant $\alpha$.

Third, we consider the case where $\lim_{k\rightarrow\infty}(k-1)a(k) = 0$. 
We focus on the Taylor series $W_0$ of $W$ around 0.
It can be found using the Lagrange inversion and is given by
$$
W_{0}(x) = \sum_{i=1}^{\infty}\frac{(-i)^{i-1}}{i!}x^{i} = x - x^{2} + O(1)x^{3}.
$$
Thus, for $k$ large  enough $g(k) = \Theta((k-1)a(k))$.  Hence, $z_{0}
=  \Theta(1/a(k))$.  Once  again that  implies
$b(k)   =  \Theta\left(\alpha^{k}\cdot\frac{1}{ka(k)^{k}}\right)$  for
some constant $\alpha$.
\end{proof}

\begin{lemma} 	\label{lem:smooth-convex}
For any sequences of non-negative real numbers 
$\{a_{1}, a_{2}, \ldots, a_{n}\}$ and $\{b_{1}, b_{2}, \ldots, b_{n}\}$
and for any polynomial $g$ of degree $k$ with non-negative coefficients, 
it holds that
\begin{align*}
\sum_{i=1}^{n} \left[ g\biggl( b_{i}  + \sum_{j=1}^{i} a_{j} \biggr) - g\biggl( \sum_{j=1}^{i} a_{j} \biggr) \right]
\leq \lambda(k) \cdot g\biggl( \sum_{i=1}^{n} b_{i}  \biggr) + 
	\mu(k) \cdot g\biggl( \sum_{i=1}^{n} a_{i}  \biggr) 
\end{align*}
where $\mu(k) = \frac{k-1}{k}$ and $\lambda(k) = \Theta\left(k^{k-1}\right)$.  
The same inequality holds for  $\mu(k) = \frac{k-1}{k \ln k}$ and $\lambda(k) = \Theta\left((k \ln k)^{k-1}\right)$.  
\end{lemma}
\begin{proof}
We first prove for $\mu(k) = \frac{k-1}{k}$ and $\lambda(k) = \Theta\left(k^{k-1}\right)$.
Let $g(z) = g_{0}z^{k} + g_{1}z^{k-1} + \cdot + g_{k}$ with $g_{t} \geq 0 ~\forall t$.
The lemma holds since it holds for every $z^{t}$ for $0 \leq t \leq k$.
Specifically,
\begin{align}	\label{eq:proof-smooth-convex}
\sum_{i=1}^{n} &\left[ g\biggl( b_{i}  + \sum_{j=1}^{i} a_{j} \biggr) - g\biggl( \sum_{j=1}^{i} a_{j} \biggr) \right]
= \sum_{t=1}^{k} g_{k - t} \cdot \sum_{i=1}^{n} \left[ \biggl( b_{i}  + \sum_{j=1}^{i} a_{j} \biggr)^{t} - \biggl( \sum_{j=1}^{i} a_{j} \biggr)^{t} \right]  \notag \\ 
&\leq  \sum_{t=1}^{k} g_{k - t} \cdot \left[ t \cdot b_{i} \cdot \biggl( b_{i} + \sum_{j=1}^{i} a_{j}  \biggr)^{t-1} \right] 
\leq  \sum_{t=1}^{k} g_{k - t} \cdot \left[ \lambda(t) \biggl( \sum_{i=1}^{n} b_{i}  \biggr)^{t} + 
	\mu(t) \biggl( \sum_{i=1}^{n} a_{i}  \biggr)^{t} \right] \notag \\ 
&\leq   \lambda(k) \cdot g\biggl( \sum_{i=1}^{n} b_{i}  \biggr) + 
	\mu(k) \cdot g\biggl( \sum_{i=1}^{n} a_{i}  \biggr)
\end{align}
The first inequality follows the convex inequality $(x+y)^{k+1} - x^{k+1} \leq (k+1)y(x+y)^{k}$.
The second inequality follows Lemma~\ref{lem:smooth-simple} (Case 2 and $a(k) = 1/(k+1)$). 
The last inequality holds since $\mu(t) \leq \mu(k)$ and $\lambda(t) \leq \lambda(k)$ for $t \leq k$.

The case  $\mu(k) = \frac{k-1}{k \ln k}$ and $\lambda(k) = \Theta\left((k \ln k)^{k-1}\right)$ is proved similarly. 
The only different step is in the second inequality of (\ref{eq:proof-smooth-convex}). 
In fact, applying Lemma~\ref{lem:smooth-simple} (Case 3 and $a(k) = \frac{1}{(k+1)\ln k}$), 
one gets the lemma inequality for $\mu(k) = \frac{k-1}{k \ln k}$ and $\lambda(k) = \Theta\left((k \ln k)^{k-1}\right)$.
\end{proof}

\newpage

\bibliographystyle{plainnat}

\end{document}